\documentclass[onecolumn,a4paper,accepted=2020-09-01]{quantumarticle}
	\pdfoutput=1
\usepackage{graphicx}
\usepackage[utf8]{inputenc}
\usepackage[english]{babel}
\usepackage{color}
\usepackage{amsmath}
\usepackage{amsthm}
\usepackage{amssymb}
\usepackage{multirow}
\usepackage{enumerate}
\usepackage{cite}
\usepackage{hyperref}

\newtheorem{Th}{Theorem}
\newtheorem{Prop}[Th]{Proposition}
\newtheorem{Cor}[Th]{Corollary}

\theoremstyle{remark}
\newtheorem{F}{Fact}
\newtheorem{claim}[F]{Fact}
\newtheorem{R}{Remark}

\theoremstyle{definition}
\newtheorem{Df}{Definition}
\newtheorem{Ex}{Example}
\newtheorem{definition}[Df]{Definition}
\def\tr{\textnormal{tr}}

\def\lin{\textnormal{lin}}
\def\conv{\textnormal{conv}}

\def\ext{\textnormal{ext}}
\def\mixed{\mathcal S(\mathbb C^d)}

\begin{document}
	
\title{Morphophoric POVMs, generalised qplexes,\newline and 2-designs}

\author{Wojciech S{\l}omczy\'{n}ski}
\email{wojciech.slomczynski@im.uj.edu.pl}
\author{Anna Szymusiak}
\email{anna.szymusiak@uj.edu.pl}
\affiliation{Institute of Mathematics, Jagiellonian University, \L ojasiewicza 6, 30-348 Krak\'{o}w, Poland}

\begin{abstract}
	We study the class of quantum measurements with the property that the image of 	the set of quantum states under the measurement map transforming states into 	probability distributions is similar to this set and call such measurements morphophoric. 	This leads to the generalisation of the notion of a qplex, where SIC-POVMs are replaced 	by the elements of the much larger class of morphophoric POVMs, containing in particular 	2-design (rank-1 and equal-trace) POVMs. The intrinsic geometry of a~generalised qplex is the same 	as that of the set of quantum states, so we explore its external geometry, 	investigating, inter alia, the algebraic and geometric form of the inner (basis) and 	the outer (primal) polytopes between which the generalised qplex is sandwiched. In 	particular, we examine generalised qplexes generated by MUB-like 2-design POVMs utilising 	their graph-theoretical properties. Moreover, we show how to extend the primal equation 	of QBism designed for SIC-POVMs to the morphophoric case.
\end{abstract}

\maketitle

\section{Introduction and preliminaries}
\label{Introduction}

Over the last ten years in a series of papers
\cite{FucSch09,FucSch11,FucSch13,Appetal11,Appetal17,Ros11,Kiketal19} Fuchs, Schack,
Appleby, Stacey, Zhu and others have introduced first the idea of \textit{QBism}
(formerly \textit{quantum Bayesianism}) and then its probabilistic embodiment - the (\textit{Hilbert})
\textit{qplex}, both based on the notion of \textit{symmetric informationally
complete positive operator-valued measure} (\textit{SIC-POVM}), representing the
standard (or reference) measurement of a quantum system. In spite
of the fact that SIC-POVMs do indeed exhibit some unique properties that
distinguish them among other IC-POVMs \cite{DeBFucSta19b,DeBFucSta20,Kwa19},
we believe that similar (or, in a~sense, even richer from a purely mathematical
point of view) structures, say, \textit{generalised qplexes}, can be
successfully used to faithfully represent quantum states as probability
distributions if we replace SIC-POVMs with the broader class
of \textit{morphophoric POVMs}, containing in particular rank-$1$ measurements
generated by \textit{complex projective} $2$\textit{-designs}. As in the QBism
interpretation of quantum mechanics, the result is a new formalism for quantum theory,
which involves only probabilities.

The rank-$1$ equal-trace case\footnote{In fact, Fuchs and Schack considered such a generalisation
in \cite[pp. 22-23]{FucSch09}, but it seems that they rejected it. See also Fuchs' letter to Zauner
\cite[p. 1933]{Fuc15}, as well as \cite{Sta19}.} is easier
to cope with than the general one, as it leads to geometry similar, though not identical, to that for qplexes, as well as to an almost identical form of the \textit{primal
equation} (\textit{`Urgleichung'}) that plays the central role in the QBism approach to quantum mechanics.
The general case, where the effects of the measurement are not necessarily rank-$1$, seems to
be more complex: while the geometric picture remains in principle almost the same, the primal equation takes
a slightly different form.

Let us imagine that the system is in a given initial state. Then either we apply a measurement \textit{`on the ground'} directly, or we first send the quantum system through the measurement device
\textit{`in the sky'}. The original `Urgleichung' allows us to express the probabilities
of the outcomes of the measurement `on the ground' in terms of the probabilities of the results
of the counterfactual measurement `in the sky' (which is a SIC-POVM in the QBism formulation)
and the conditional, also counterfactual, probabilities of obtaining
a given outcome of the former measurement conditioned on receiving a given result of the latter,
provided the state of the system after performing the latter is given by the L\"{u}ders instrument.
Note, however, that in the general (not rank-$1$) case these conditional probabilities may depend on the initial state
of the system \cite[p. 23]{FucSch09} and, therefore, the primal equation cannot be easily extended to this case. That is why we
propose another equation connecting those probabilities which is equivalent to the primal equation
in the rank-$1$ equal-trace case, but employs the conditional probabilities defined for the maximally mixed initial
state rather than those for the original initial state for which the primal equation is formulated.
Despite this change, the new primal equation is also defined in purely probabilistic terms,
although its form does not resemble any of the known laws of classical probability.

\medskip

To see what path leads to this generalisation, let us start from several technicalities.
In standard (finite) $d$-dimensional quantum mechanics the states of a system
are identified with the set of positive trace-one operators $\mathcal{S}(\mathbb{C}^{d})$
(also known as density operators or density matrices) and the pure states with
the extreme boundary of this set, $\mathcal{P}(\mathbb{C}^{d})$,
i.e. the rank-1 orthogonal projections or, equivalently,
with the elements of the projective complex vector space $\mathbb{CP}^{d-1}$.
The quantum states can be also described as the elements of the ($d^{2}-1$)-dimensional
real Hilbert space $\mathcal{L}_{s}^{0}\left(\mathbb{C}^{d}\right)$
of Hermitian traceless operators on $\mathbb{C}^{d}$ endowed
with the Hilbert-Schmidt product. Namely, the map  $\mathcal{S}\left(
\mathcal{H}\right)  \ni\rho\rightarrow\rho-I/d\in\mathcal{L}_{s}^{0}\left(
\mathbb{C}^{d}\right)$ gives an affine embedding of the set of states
(resp. pure states) into the (\textit{outer}) ($d^{2}-1$)-dimensional ball
(resp. sphere) in $\mathcal{L}_{s}^{0}\left( \mathbb{C}^{d}\right)  $
centred at $0$ of radius $\sqrt{1-1/d}$. The image of this map is called
the \textit{Bloch body} (resp. the \textit{generalised Bloch sphere}). Only
for $d=2$ the map is onto, and for $d>2$ its image is a convex subset of the
ball of full dimension and contains the maximal (\textit{inner}) ball of
radius $1/\sqrt{d\left( d-1\right)}$ \cite{BenZyc17}.

A measurement with a finite number $n$ of possible outcomes can be described
by a \textit{positive operator-valued measure} (\textit{POVM}), i.e. an
ensemble of positive non-zero operators $\Pi=(\Pi_{j})_{j=1}^n$
on $\mathbb{C}^{d}$ that sum to the identity operator, i.e.
$\sum\nolimits_{j=1}^{n}\Pi_{j}=I$. According to the \textit{Born rule},
the affine \textit{measurement map} $p_{\Pi}$ sends an input quantum
state $\rho\in\mathcal{S}(\mathbb{C}^{d})  $ into the vector of
probabilities of the respective outcomes $(\operatorname{tr}\left(  \rho
\Pi_{j}\right)  )_{j=1}^n$ that belongs to the \textit{standard}
(\textit{probability}) \textit{simplex} $\Delta_{n}:=
\{x\in \left( \mathbb{R}^{+}\right)^{n}:\sum\nolimits_{j=1}^{n}x_{j}=1\}$.
If this map is one-to-one, we call $\Pi$
\textit{informationally complete}; then, necessarily, $n\geq d^{2}$. In this
situation the results of the measurement  uniquely determine the input state,
and so \textit{quantum tomography}, i.e. the task of reconstructing the quantum state
from the outcomes of the measurement, can be concluded. In this case
$\mathcal{Q}_{\Pi}:=p_{\Pi}(\mathcal{S}(  \mathbb{C}^{d})  )$ (the
\textit{probability range} of $\Pi$) is an affine ($d^{2}-1$)-dimensional
image of $\mathcal{S}(\mathbb{C}^{d})$ and so, clearly, the set of extreme
points of the closed convex set $\mathcal{Q}_\Pi$, i.e. its extreme boundary
$\ext\mathcal{Q}_{\Pi}$, is equal to $p_{\Pi}(\mathcal{P}(  \mathbb{C}^{d})  )$.
We shall analyse the structure of the image of $p_{\Pi}$ in Sections \ref{The probability range of POVM}
and \ref{The probability range similar to quantum state space}. If $\mathcal{S}(\mathbb{C}^{d})$ and
$\mathcal{Q}_{\Pi}$ are of the same shape (geometrically similar), we call such $\Pi$ \textit{morphophoric}
(or form-bearing, from Greek $\mu o\rho\phi\acute{\eta}$ - shape, form and $\phi \acute{o}\varrho o\sigma$ - bearing).
In this case we call $\mathcal{Q}_{\Pi}$ a \textit{generalised (Hilbert) qplex.}
It seems that morphophoric POVMs share with SIC-POVMs many properties crucial for developing natural generalisations of QBism and qplexes.

The main result of the present paper, Theorem \ref{thmsim}, provides a necessary and sufficient condition for
morphophoricity: the orthogonal projections onto $\mathcal L_s^0(\mathbb C^d)$ of the effects constituting
the POVM in question need to form a tight operator frame in this space. In consequence, we can
identify the set of quantum states with a convex subset of $\Delta_{n}$. In Section
\ref{The probability range similar to quantum state space} we take a closer look at the structure
of the class of morphophoric POVMs. In particular, we show that this notion coincides with that
of tight IC-POVMs introduced by Scott in \cite{Scott06} for POVMs consisting of effects of equal trace
(Corollary \ref{Scott}), but in general these classes are different (Example \ref{Scott_example}).

In this paper we pay special attention to POVMs $\Pi=(\Pi_{j})_{j=1}^n$ such that the effects $\Pi_{j}$ ($j=1,\ldots,n$)
are not only of equal trace but also one-dimensional, and so proportional to certain pure
states $\rho_{j}$ ($j=1,\ldots,n$). We show that in this case $\Pi$ is morphophoric if and only if these states
constitute a~complex projective $2$-design (Corollary \ref{cor2des}), or, equivalently, their image in the Bloch
body under the natural embedding is a tight frame in $\mathcal{L}_{s}^{0}\left(\mathbb{C}^{d}\right)$.
Note that the class of $2$-designs contains, besides SIC-POVMs \cite{Fucetal17},
also complete sets of MUBs (mutually unbiased bases)
\cite{KlaRot05}, MUB-like designs \cite{HugWal}, orbits of the multiqubit Clifford group \cite{Zhu17} (in fact,
the latter are $3$-designs), group designs \cite{Groetal07},
and many other discrete objects representing quantum measurements \cite{BenZyc17,BenZyc17b,BodHaa17}.
Let us now take a closer look at the geometry of $2$-design POVMs.

The main difference between (Hilbert) qplexes and their generalisation lies not in their
intrinsic geometry, but rather in their external geometry (i.e. the location in the underlying
real vector space), which we describe in detail in Section \ref{Geometric properties of generalised qplex}.
Both a (Hilbert) qplex and a (Hilbert) generalised qplex are isomorphic (similar) through the measurement map
to the space of quantum states (or, to put it another way, the Bloch body).
The former, however, is a subset of the standard $(d^{2}-1)$-dimensional
simplex $\Delta_{d^{2}}$ lying in a $d^{2}$-dimensional real vector space,
whilst the latter is a subset of a $(d^{2}-1)$-dimensional \textit{primal polytope} $\Delta$,
which is a $(d^2-1)$\textit{-section} of the `larger' simplex $\Delta_{n}$.
We show that the section in question has to be \textit{central}, i.e. it passes through the centre of the simplex $c_{n}$, and, what is more, \textit{medial}, i.e.
the vertices of the simplex are equidistant from the \textit{primal affine space}
$\mathcal{A}$ generated by $\Delta$. (The term `medial' in the general case, i.e. not restricted
to the hyperplane sections, has been introduced in \cite{Tal03}.)
The distance equal to $1-d^{2}/n$ measures the deviation from
minimality: for $n=d^{2}$ we get necessarily a SIC-POVM. Moreover, we consider
another polytope, namely a \textit{basis polytope} $D$ generated by the elements
of a $2$-design transformed by $p_{\Pi}$. This polytope is the image
under the homothety (with centre at $c_{n}$ and ratio  $1/(d+1)$)
of the orthogonal projection of $\Delta_{n}$ on $\mathcal{A}$.

The generalised qplex $\mathcal{Q}_{\Pi}$
is a convex set sandwiched between these two polytopes $D$ and $\Delta$. The situation is
somewhat analogous to the quantum (Bell) correlation picture, where the quantum
convex body lies in-between the \textit{local polytope} \cite{Pit86} and the
\textit{non-signalling polytope} \cite{PopRoh94}, see also
\cite{Bub06,Fri12,Brunneretal14,GelPia14,Pop14,Gohetal18}, or to the related
idea of applying graph theoretical concepts \cite{Groetal93,Knu94} to the
analysis of quantum contextuality, where the probabilities of the results of an experiment
in classical, quantum, and more general probabilistic theories are described
by, respectively, \textit{stable set polytope} $\subset $ \textit{theta body} $\subset $
\textit{fractional stable set polytope} (or, in other words, \textit{clique-constrained stable
set polytope}) of the graph related to the experiment \cite{Cabetal10,Cabetal14,AmaCun18}.
In fact, the analogy is even deeper since also in our situation $D$
and $\Delta$ represent, in a sense,
classical and `beyond-quantum' behaviour, respectively.

The polytopes $D$ and $\Delta$ are \textit{dual }in $\mathcal{A}$, see \cite{Wei11},
namely $D$ and $\iota(\Delta)$, where $\iota$ is the inversion through $c_{n}$, are
\textit{polar} with respect to the central sphere of radius $m := 1/\sqrt{n(d+1)}$.
Moreover, $\mathcal{Q}_\Pi$ is \textit{self-dual} in this sense,
every pair of elements $p,q\in\mathcal{Q}_\Pi$ fulfils the fundamental
inequalities: $dm^{2}\leq \langle p, q\rangle\leq2dm^{2}$, and $\mathcal{Q}_\Pi$ lies
in-between two central balls: the inscribed ball of $\Delta$ of radius
$r:=m/\sqrt{d-1}$ and the circumscribed ball of $D$ of radius
$R:=m\sqrt{d-1}$. Note that $m=\sqrt{rR}$.
Thus, the overall geometric picture is similar to the (Hilbert) qplex
case represented schematically in \cite[Figs. 2 and 3]{Appetal17}, but more complicated
as our polytopes are not necessarily simplices, see Figs.~\ref{Fig4} and \ref{Fig5}.
These geometric properties, especially the mediality of the constitutive cross-section
space $\mathcal{A}$, can be chosen as a starting point for an abstract definition of
generalised qplexes (in preparation).
Note that most of these properties hold true for every morphophoric POVM, as we see at the end
of Section \ref{Geometric properties of generalised qplex}.

This geometric approach can be supplemented by an algebraic one we present
in Section \ref{The probability range of 2-design POVMs}, namely, for
the extreme boundary $\ext\mathcal{Q}_\Pi$ of the qplex, we show that
it can be described by a set of polynomial equations (Theorem \ref{thm2des}): $n$ linear,
one quadratic, and one cubic (for $d > 2$). The system of $n$ linear equations defining the affine
space $\mathcal{A}$ is just an analogue of the famous primal equation (`Urgleichung') known
from the QBism theory where $\Pi$ plays the role of both the `ground' and the `sky' measurement.
Note, however, that for SIC-POVMs these equations do not impose any constraints on probabilities as they
are always satisfied (Corollary \ref{Special}.(ii)).
For low-dimensional systems ($d=2,3$) we provide also the description of the whole generalised
qplex $\mathcal{Q}_\Pi$ with the help of a set of polynomial equations and inequalities (Corollary
\ref{ineq2des}).

We provide a detailed analysis of several examples of
$2$-designs and the corresponding generalised qplexes: cubical,  cuboctahedral
and icosahedral POVMs in dimension $2$, as well as SIC-POVMs in dimension $3$ in Section \ref{Examples}.
 The linear equations defining primal affine space
$\mathcal{A}$ take an unexpectedly simple form for complete sets of MUBs and MUB-like 2-designs (i.e. two-distance
2-designs with one of inner products equal to $0$), see Section \ref{MUB-like POVMs}.
Apart from complete sets of MUBs, we know four  sporadic objects of this type, one in each of dimensions
4, 5, 6 and 28, and they all generate the orthogonality
graph on the elements of the POVM that is a strongly regular graph with some special properties.
In this case the equations describing $\mathcal{A}$ can be replaced by the requirement
that \textit{the sum of the probabilities over each orthonormal basis is constant}
(Theorem~\ref{Bases}).

In the last section we show that a POVM is morphophoric if and only if it obeys an equation which is a natural generalisation of the quantum law of total probability of QBism (Theorem \ref{morpheq} and Corollary \ref{morpheqcor}).
Namely, for a morphophoric measurement $\Pi=(\Pi_{j})_{j=1}^{n}$ `in the sky' and an arbitrary measurement
$\Xi=(\Xi_{k})_{k=1}^{N}$ `on the ground' we define the \textsl{deviations} of the
probability distributions of the measurements results for the pre-measurement
state $\rho$ from those for $\rho_{\ast}$, where $\rho_{\ast} := I/d$ is the
maximally mixed state: $\delta_{\Pi}(\rho):=p_{\Pi}(\rho)-p_{\Pi}(\rho_{\ast})$
and $\delta_{\Xi}(\rho):=p_{\Xi}(\rho)-p_{\Xi}(\rho_{\ast})$. Moreover,
we consider the transposed covariance matrix $C$ of two consecutive measurements $\Pi$
and $\Xi$ given by $C_{kj} :=p_{jk}^{\Pi\Xi}(\rho_{\ast})-p_{j}^{\Pi}(
\rho_{\ast}) p_{k}^{\Xi}(  \rho_{\ast})$, for $j=1,\ldots,n$ and $k=1,\ldots,N$,
and the constant
$\zeta := (d^2-1)^{-1}\sum_{j=1}^{n}(  p_{jj}^{\Pi\Pi}(\rho_{\ast})-(p_{j}^{\Pi}(\rho_{\ast}))^{2})$.
Then the following simple generalisation of the primal equation
$$\zeta \delta_{\Xi} = C\delta_{\Pi}$$
links the probabilities of both measurements.

\section{The probability range of POVMs}
\label{The probability range of POVM}

Let us consider a quantum system represented by $\mathbb C^d$ and a general discrete quantum
measurement on it  described by a POVM (positive operator-valued measure) $\Pi=(\Pi_j)_{j=1}^n$.
In particular, if  $\Pi_j$ are orthogonal projections, we call $\Pi$ a PVM (projection-valued measure). A POVM is called informationally complete (IC-POVM) if its statistics determine uniquely the pre-measurement state, i.e. the conditions $\tr(\rho\Pi_j)=\tr(\sigma\Pi_j)$ ($j=1,\ldots,n$) imply $\rho=\sigma$ for every input states $\rho,\sigma\in\mathcal S(\mathbb C^d)$.
Let us denote by $\mathcal{Q}_\Pi$ the set of `allowed' probabilities, i.e. the set of all possible
probability distributions of the measurement outcomes over all quantum states. It is known also as
the \emph{probability range} of $\Pi$ \cite{Heietal19}. More formally, if we denote by $\mathcal L_s(\mathbb C^d)$
the $d^2$-dimensional real ordered vector space of Hermitian operators on $\mathbb C^d$, then
$$\mathcal{Q}_\Pi:=p_\Pi(\mathcal S(\mathbb C^d)),$$ where
$$p_\Pi:\mathcal L_s(\mathbb C^d)\ni A\mapsto (\tr(A\Pi_1),\ldots,\tr(A\Pi_n))\in\mathbb R^n$$ is a positive linear map.
We introduce also $\mathcal{Q}_\Pi^1$ for the probability distributions attained for the pure states, i.e.
$$\mathcal {Q}_\Pi^1:=p_\Pi(\mathcal P(\mathbb C^d)).$$
As a straightforward consequence, we get the following properties of $\mathcal {Q}_\Pi$ and $\mathcal {Q}_\Pi^1$:
\begin{F}
$\mathcal {Q}_\Pi=\conv(\mathcal {Q}_\Pi^1)$ and
$\ext(\mathcal {Q}_\Pi)\subset \mathcal {Q}_\Pi^1$.
\end{F}
Let us recall that the \emph{joint numerical range} of Hermitian matrices $A_1,A_2,\ldots,A_n$ is
defined by
$$\mathcal F(A_1,A_2,\ldots,A_n)=\{(\langle z|A_1|z\rangle,\langle z|A_2|z\rangle,\ldots,\langle z|A_n|z\rangle)^T|z\in\mathbb C^d, \|z\|=1\}\subset\mathbb R^n.$$
As any pure state $\rho$ is necessarily the orthogonal projection onto some unit vector $z\in\mathbb C^d$,
we have $\tr(\rho\Pi_j)=\langle z|\Pi_j|z\rangle$, and so $\mathcal {Q}_\Pi^1$ can be described also
as the joint numerical range of $\Pi_1,\Pi_2,\ldots,\Pi_n$.

We are particularly interested in the shape of $\mathcal {Q}_\Pi$. For $d=2$, $\mathcal {Q}_\Pi$ is an ellipsoid,
possibly degenerated, as a linear image of 3-dimensional ball. Some examples of these ellipsoids are presented below.
\begin{Ex} Let us consider a one-parameter family of rank-1 normalised POVMs on $\mathbb C^2$ consisting of $n=4$ effects:
$$\Pi_1:=\frac{1}{4}\begin{pmatrix} 1+\sin\alpha & \cos\alpha\\
\cos\alpha & 1-\sin\alpha
\end{pmatrix},\ \Pi_2:=\frac{1}{4}\begin{pmatrix} 1+\sin\alpha & -\cos\alpha\\
-\cos\alpha & 1-\sin\alpha
\end{pmatrix},$$$$
\Pi_3:=\frac{1}{4}\begin{pmatrix} 1-\sin\alpha & -i\cos\alpha\\
i\cos\alpha & 1+\sin\alpha
\end{pmatrix},\ \Pi_4:=\frac{1}{4}\begin{pmatrix} 1-\sin\alpha & i\cos\alpha\\
-i\cos\alpha & 1+\sin\alpha
\end{pmatrix},
$$ where $\alpha\in[0,\pi/2]$. The pure states $\rho_j=2\Pi_j$ for $j=1,\ldots,4$ are represented on the Bloch sphere
by the vertices of certain tetragonal disphenoid, i.e.\  a tetrahedron whose four faces are congruent isosceles triangles.
Such a tetrahedron can be uniquely characterised by the length of the base (or legs) of that triangle. Taking into account
possible degenerations, we get five cases. For $\alpha=\frac{\pi}{2}$ the disphenoid degenerates to a segment (length of
the base $= 0$) and the measurement effectively  acts as the  projective von Neumann measurement. For $\alpha\in(0,\pi/2)$
there is no degeneracy and the POVM is informationally complete. In particular, for $\alpha\in(\arcsin\frac{1}{\sqrt{3}},\frac{\pi}{2})$
the length of the base of the face triangle is less than the length of its legs. For $\alpha=\arcsin\frac{1}{\sqrt{3}}$ we get the
tetrahedron (faces being equilateral triangles) representing SIC-POVM, the only case of 2-design here.  For
$\alpha\in(0,\arcsin\frac{1}{\sqrt{3}})$ the disphenoid is such that the length of the base of the face triangle
is greater than the length of its legs. Finally, for $\alpha=0$ the disphenoid degenerates to a square (length of the legs $=0$)
representing a non-informationally complete highly symmetric POVM. Its probability range in these cases is presented in Fig.\ \ref{sphs}
and evolves as follows: for $\alpha=\frac{\pi}{2}$ we obtain  a segment, then for $\alpha\in(\arcsin\frac{1}{\sqrt{3}},\pi/2)$ -- an elongated spheroid, for
$\alpha=\arcsin\frac{1}{\sqrt{3}}$ -- a~ball, for $\alpha\in(0,\arcsin\frac{1}{\sqrt{3}})$ -- a~flattened spheroid and, finally, for
$\alpha=0$ -- a disk.
\begin{figure}[htb]	
\begin{center}\includegraphics[scale=0.078]{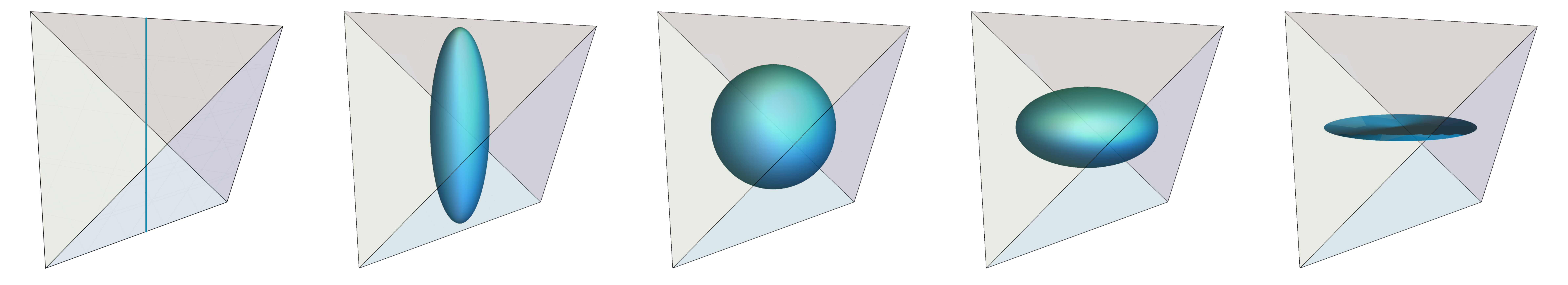}
\end{center}
\caption{$\mathcal {Q}_\Pi$ in the cases: $\alpha=\frac{\pi}{2}$ -- a segment; $\alpha\in(\arcsin\frac{1}{\sqrt{3}},\pi/2)$ -- an elongated spheroid; $\alpha=\arcsin\frac{1}{\sqrt{3}}$ -- a ball; $\alpha\in(0,\arcsin\frac{1}{\sqrt{3}})$ -- a flattened spheroid; $\alpha=0$ -- a disk. All solids are tangent to $\Delta_4$.}\label{sphs}
\end{figure}
\end{Ex}

It is much more difficult to describe the shape of $\mathcal {Q}_\Pi$ in dimensions higher than 2. However, it seems that there are two cases which are in some sense `nice'. The first one is when $\mathcal {Q}_\Pi$ is the full simplex. The second one is when it is of the same shape as $\mathcal S(\mathbb C^d)$. In what follows we deal with the sufficient and necessary conditions for these two to appear.
As for the full simplex case, the following characterisation  may be categorised as quantum information folklore \cite[p. 4]{Heietal19}, see below. The second case will be thoroughly analysed in the next section.

\begin{Df}
We say that a POVM $\Pi$ has the \textit{norm-1-property} if $\|\Pi_j\|=1$ for $j=1,\ldots,n$ (here $\|A\|:=\max_{\|x\|_2=1}\|Ax\|_2$).
\end{Df}

\begin{Th}
$\mathcal {Q}_\Pi=\Delta_n$ if and only if $\Pi$ has the norm-1-property. Moreover, in such case, $\mathcal {Q}_\Pi^1=\Delta_n$, $n\leq d$, and $n=d$ if and only if $\Pi$ is a rank-1 PVM.
\end{Th}

\section{The probability range similar to quantum state space}
\label{The probability range similar to quantum state space}

In order to provide a characterisation of the second case we need to make some preparations.
From  now on for $A,B\in\mathcal L_s(\mathbb C^d)$ we will denote their Hilbert-Schmidt inner product by $(A|B):=\tr(AB)$ and the induced norm of $A$ by $\|A\|_{HS}$. We also consider
 $\mathcal L_s^0(\mathbb C^d):=\{A\in\mathcal L_s(\mathbb C^d)|\tr A=0\}$, i.e.\ the ($d^2-1$)-dimensional subspace of traceless Hermitian operators. It is easy to see that
\begin{equation}\label{pi0}\pi_0:\mathcal L_s(\mathbb C^d)\ni A\mapsto A-(\tr A/d)I\in\mathcal L_s(\mathbb C^d)\end{equation}
is the orthogonal projection onto $\mathcal L_s^0(\mathbb C^d)$.

 \begin{Prop}\label{IC}The following conditions are equivalent:
 	\begin{enumerate}[i.]
 		\item $\Pi=(\Pi_j)_{j=1}^n$ is an IC-POVM;
 		\item $\lin(\Pi):=\lin\{\Pi_j:j=1,\ldots,n\}=\mathcal L_s(\mathbb C^d)$;
 		\item $\lin(\pi_0(\Pi)):=\lin\{\pi_0(\Pi_j):j=1,\ldots,n\}=\mathcal L_s^0(\mathbb C^d)$.
 	\end{enumerate}
 	 \end{Prop}
 \begin{proof}
 	The equivalence of (i) and (ii) can be found in \cite[Prop. 3.51]{HeiZim12}. In order to see the equivalence of (ii) and (iii), let us first observe that $(\lin(\Pi))^\perp\subset\mathcal L_s^0(\mathbb C^d)$ since $I\in\lin(\Pi)$. 	Thus
 \begin{align*}\mathcal L_s^0\cap(\lin(\pi_0(\Pi)))^\perp&=\{A\in\mathcal L_s^0:(A|\Pi_j)-\tr A\tr\Pi_j/d=0, j=1,\ldots,n\}\\
 	&=\{A\in\mathcal L_s^0:(A|\Pi_j)=0, j=1,\ldots,n\}=\mathcal L_s^0\cap(\lin(\Pi))^\perp=\mathcal L_s\cap(\lin(\Pi))^\perp. \qedhere
 \end{align*}
 \end{proof}
  The following definition gives a precise description of what we mean by `the same shape' in terms of similarity. Throughout the paper we use $\langle\cdot,\cdot\rangle$ and $\|\cdot\|$ to denote the standard inner product  and the induced norm on $\mathbb R^n$.
\begin{Df}
We say that $\mathcal {Q}_\Pi$ is \textit{similar} to quantum state space $\mathcal S(\mathbb C^d)$ if there exists $\alpha>0$ such that
\begin{equation}
\|p_\Pi(\rho)-p_\Pi(\sigma)\|^2=\alpha\|\rho-\sigma\|_{HS}^2\quad \textnormal{for all }\rho,\sigma\in\mathcal S(\mathbb C^d).
\end{equation}
\end{Df}
\begin{R} The above definition can be expressed equivalently as
\begin{equation}\label{iso}\langle p_\Pi(\rho)-c_\Pi,p_\Pi(\sigma)-c_\Pi\rangle=\alpha(\rho-\rho_{\ast}|\sigma-\rho_{\ast})
\end{equation}
for all $\rho,\sigma\in\mathcal S(\mathbb C^d)$, where $c_\Pi:=p_\Pi(\rho_{\ast})=(\tr\Pi_1/d,\ldots,\tr\Pi_n/d)$. (Note that $\rho-\rho_{\ast}=\pi_0(\rho)$ and $p_\Pi(\rho)-c_\Pi=p_\Pi(\pi_0(\rho))$.) In particular, $p_\Pi \vert_{\mathcal L_s^0(\mathbb C^d)}$ is a similarity with the similarity ratio $s=\sqrt{\alpha}$.
\end{R}
\begin{Df}
	We say that a POVM $\Pi$ is \emph{morphophoric} if $\mathcal Q_\Pi$ is similar to the quantum state space $\mathcal S(\mathbb C^d)$.	
\end{Df}
Let us recall some basic facts about finite tight frames. Let $V$ be a finite-dimensional Hilbert space with inner product $\langle\cdot|\cdot\rangle$ and let
$F:=\{f_1,\ldots,f_m\}\subset V$.
\begin{Df}	
	$F$ is a \textit{tight frame} if there exists $\alpha>0$ such that $$\alpha\|v\|^2=\sum_{i=1}^m|\langle v|f_i\rangle|^2\quad \textnormal{for all } v\in V.$$
	In such situation $\alpha$ is referred to as the \textit{frame bound}.
	The operator $S:=\sum_{i=1}^m|f_i\rangle\langle f_i|$ is called the \textit{frame operator}.
\end{Df}
The next theorem is a well-known fact in frame theory (see, e.g. \cite{Waldron}) that provides us with some other handful equalities characterising tight frames.


\begin{Th}\label{tf}The following conditions are equivalent:
\begin{enumerate}[i.]
\item $F$ is a tight frame with frame constant $\alpha$;
\item $\sum_{i=1}^m\langle u|f_i\rangle\langle f_i|v\rangle=\alpha \langle u|v\rangle$ for all $u,v\in V$;
\item $S=\alpha I$.
\end{enumerate}
\end{Th}
Now we can proceed to the characterisation of morphophoric POVMs.
\begin{Th}\label{thmsim}
Let $\Pi$ be a POVM. Then $\mathcal {Q}_\Pi$ is similar to the quantum state space if and only if $\pi_0(\Pi)$ is a~tight (operator) frame in $\mathcal L_s^0(\mathbb C^d)$. Moreover, the frame bound and the square of the similarity ratio  coincide and are equal to \begin{equation}\label{squaresim}\alpha=\frac{1}{d^2-1}\left(\sum_{j=1}^n\tr\Pi_j^2-\frac{1}{d}\sum_{j=1}^n(\tr\Pi_j)^2\right).\end{equation}
\end{Th}

\begin{proof}
Since $\textnormal{span}\{\rho-\rho_{\ast}:\rho\in\mixed\}=\mathcal L_s^0(\mathbb C^d)$, it follows that  \eqref{iso} holds for all quantum states if and only if it holds for all $\rho,\sigma\in\mathcal L_s(\mathbb C^d)$ such that $\tr\rho=\tr\sigma=1$. Thus, by  linearity, we can extend this condition to arbitrary Hermitian operators by writing it equivalently  as
\begin{equation}\label{simall}\langle p_\Pi(A)-\tr Ac_\Pi,p_\Pi(B)-\tr Bc_\Pi\rangle=\alpha( A-(\tr A/d)I|B-(\tr B/d)I)\quad \textnormal{for all }A,B\in\mathcal L_s(\mathbb C^d).
\end{equation} Transforming the lhs, we get
\begin{align*}\langle p_\Pi(A)&-\tr Ac_\Pi,p_\Pi(B)-\tr Bc_\Pi\rangle=\sum_{j=1}^n(\tr(A\Pi_j)-\tr A\tr\Pi_j/d)(\tr(B\Pi_j)-\tr B\tr\Pi_j/d)\\
&=\sum_{j=1}^n\tr[(A-(\tr A/d)I)(\Pi_j-(\tr\Pi_j/d)I)]\tr[(B-(\tr B/d)I)(\Pi_j-(\tr\Pi_j/d)I)]\\
&=\sum_{j=1}^n(A-(\tr A/d)I|\pi_0(\Pi_j))(\pi_0(\Pi_j)|B-(\tr B/d)I)\\
&=(A-(\tr A/d)I|S|B-(\tr B/d)I),
\end{align*} where $S:=\sum_{j=1}^n|\pi_0(\Pi_j))(\pi_0(\Pi_j)|$. Thus,  \eqref{simall} is equivalent to
\begin{equation}
(A|S|B)=\alpha(A|B) \quad \text{for all } A,B\in\mathcal L_s^0(\mathbb C^d),
\end{equation}
which is equivalent to $\pi_0(\Pi)$ being a tight operator frame in $\mathcal L_s^0(\mathbb C^d)$ with the frame bound $\alpha$. In order to calculate $\alpha$, let us take the trace on both sides of the equation
\begin{equation}
\sum_{j=1}^n|\pi_0(\Pi_j))(\pi_0(\Pi_j)|=\alpha\mathcal I_0,
\end{equation}
where $\mathcal I_0$ denotes the identity superoperator on $\mathcal L_s^0(\mathbb C^d)$.
On the rhs we get $\alpha(d^2-1)$, since $\dim\mathcal L_s^0(\mathbb C^d)=d^2-1$. On the lhs we get
\begin{align*}
\sum_{j=1}^n(\pi_0(\Pi_j)|\pi_0(\Pi_j))&=\sum_{j=1}^n\left[(\Pi_j|\Pi_j)-\frac{\tr\Pi_j}{d}(I|\Pi_j)-\frac{\tr\Pi_j}{d}(\Pi_j|I)+\frac{(\tr\Pi_j)^2}{d^2}(I|I)\right]\\
&=\sum_{j=1}^n\tr\Pi_j^2-\frac{1}{d}\sum_{j=1}^n(\tr\Pi_j)^2,
\end{align*}
which completes the proof.
\end{proof}
It is easy to see that if $\Pi^1,\Pi^2,\ldots,\Pi^m$ are morphophoric POVMs with squares of the similarity ratios equal to $\alpha_1,\alpha_2,\ldots,\alpha_m$, then also $\Pi:=(t_1\Pi^1)\cup(t_2\Pi^2)\cup \ldots\cup (t_m\Pi^m)$ is a morphophoric POVM for any $t_1,t_2,\ldots,t_m\geq0$ such that $t_1+t_2+\ldots+t_m=1$. In such case the square of the similarity ratio for $\Pi$ is equal to $\alpha=t_1\alpha_1+t_2\alpha_2+\ldots+t_m\alpha_m$.

The following propositions not only give us an insight into the structure of the set of morphophoric POVMs on $\mathbb C^d$ with a fixed number of elements $n$ but also provide us with  examples of such POVMs. The first one states that an addition of classical noise  $q\mathbb I=(q_j I)_{j=1}^n$, where $q\in\Delta_n$, preserves the morphophoricity.
\begin{Prop}
	Let $\Pi=(\Pi_j)_{j=1}^n$ be a morphophoric POVM and let $q\in\Delta_n$.  Then $\Pi_{\lambda,q}:=\lambda\Pi+(1-\lambda)q \mathbb I$ is also a morphophoric POVM for $\lambda\in(0,1]$. In such case, the square of the  similarity ratio for $\Pi_{\lambda,q}$	is equal to $\lambda^2\alpha$, where $\alpha$ is the square of the similarity ratio for $\Pi$.
\end{Prop}
\begin{proof}
	It is enough to observe that $\mathcal Q_{\lambda\Pi+(1-\lambda)q\mathbb I}=\lambda\mathcal Q_\Pi+(1-\lambda)q$, which is a homothety (therefore a~similarity) with the centre at $q$ and the ratio equal to $\lambda$.		
\end{proof}
\begin{Df}
	We say that a POVM $\Pi$ is \emph{boundary} if   $\Pi_j$ has at least one eigenvalue equal to zero for every $j\in\{1,\ldots,n\}$.
\end{Df}

The next proposition states that every morphophoric POVM is a unique convex combination of a~boundary morphophoric POVM with some classical noise.
\begin{Prop}
	Let $\Pi$ be a morphophoric POVM which is not boundary. Then there exists a~unique boundary morphophoric POVM $E=(E_j)_{j=1}^n$ such that $\Pi=\lambda E+(1-\lambda)q\mathbb I$ for some $\lambda\in(0,1)$ and $q\in\Delta_n$.	
\end{Prop}
\begin{proof}
	Let $0\leq\lambda_{j,1}\leq...\leq\lambda_{j,d}$ be the eigenvalues of $\Pi_j$. As $\Pi$ is not boundary,  $\sum_{j=1}^n\lambda_{j,1}>0$. Moreover, from the fact that $\sum_{j=1}^n\sum_{k=1}^d\lambda_{j,k}=\sum_{j=1}^n\tr\Pi_j=d$ and $\Pi\neq (I/n)_{j=1}^n$ we obtain $\sum_{j=1}^n\lambda_{j,1}<1$. Let $q_j:=\lambda_{j,1}/(\sum_{j=1}^n\lambda_{j,1})$ and $\lambda := 1-\sum_{j=1}^n\lambda_{j,1}$. Then $q\in\Delta_n$ and $\lambda\in(0,1)$. Put $E_j:=(\Pi_j+(\lambda-1)q_jI)/\lambda$. Then the eigenvalues of $E_j$ are of the form $\mu_{j,k}:=(\lambda_{j,k}+(\lambda-1)q_j)/\lambda=(\lambda_{j,k}-\lambda_{j,1})/\lambda\geq 0$ and so $\mu_{j,1}=0$. Clearly, $\sum_{j=1}^nE_j=I$ and $\pi_0(E_j)=\pi_0(\Pi_j)/\lambda$. Thus, $E=(E_j)_{j=1}^n$ is a boundary morphophoric POVM.
	
	To prove uniqueness, let us assume that there exist a boundary morphophoric POVM $E'$, $\lambda'\in(0,1)$ and $q'\in\Delta_n$ such that $\Pi=\lambda'E'+(1-\lambda')q'\mathbb I$. Let $\mu_{j,k}':=(\lambda_{j,k}+(\lambda'-1)q'_j)/\lambda'$ denote the $k$-th  eigenvalue of $E_j'$. Then $\mu_{j,1}'=0=\mu_{j,1}$ and, in consequence, $(\lambda'-1)q_j'=(\lambda-1)q_j$ for  $j=1,\ldots,n$. Summing over $j$, we obtain $\lambda'=\lambda$, and so $q_j'=q_j$ for $j=1,\ldots,n$. 	
\end{proof}
A similar reasoning provides us with a \emph{dual} boundary morphophoric POVM:
\begin{Prop}
	Let $\Pi$ be a morphophoric POVM. Then there exists a unique boundary morphophoric POVM $\tilde{E}$ such that $\lambda \Pi+(1-\lambda)\widetilde{E}=q\mathbb I$ for some $\lambda\in(0,1)$ and $q\in\Delta_n$.
\end{Prop}
\begin{proof}
	Let $0\leq\lambda_{j,1}\leq\ldots\leq\lambda_{j,d}$ be the eigenvalues of $\Pi_j$. Let $q_j:=\lambda_{j,d}/(\sum_{j=1}^n\lambda_{j,d})$ and $\lambda:=1/(\sum_{j=1}^n\lambda_{j,d})$. Then $q\in\Delta_n$ and $\lambda\in (0,1)$ since $\sum_{j=1}^n\lambda_{j,d}\geq 1$. Put  $\widetilde{E_j}:=(1-\lambda)^{-1}(-\lambda \Pi_j+q_jI)$. The eigenvalues of $\widetilde{E_j}$ are of the form $\widetilde{\lambda_{j,k}}:=(-\lambda\lambda_{j,k}+q_j)/(1-\lambda)=(\lambda/(1-\lambda))(\lambda_{j,d}-\lambda_{j,k})\geq 0$ and $\widetilde{\lambda_{j,d}}=0$. Obviously, $\sum_{j=1}^n\widetilde{E_j}=I$ and $\pi_0(\widetilde{E_j})=\frac{-\lambda}{1-\lambda}\pi_0(E_j)$. Thus, $\widetilde{E}=(\widetilde{E_j})_{j=1}^n$ is a boundary morphophoric POVM.
	
	To prove uniqueness, we proceed in the same way as in the previous proposition. Let us assume that there exist a boundary morphophoric POVM $\widetilde{E'}$, $\lambda'\in(0,1)$ and $q'\in\Delta_n$ such that $\lambda'\Pi+(1-\lambda')\widetilde{E'}=q'\mathbb I$. Let $\widetilde{\lambda_{j,k}'}:= (-\lambda'\lambda_{j,k}+q'_j)/(1-\lambda')$ denote the $k$-th  eigenvalue of $\widetilde{E_j'}$. Then $\widetilde{\lambda_{j,d}'}=0=\widetilde{\lambda_{j,d}}$ and, in consequence, $q_j'/\lambda'=q_j/\lambda$ for  $j=1,\ldots,n$. Summing over $j$, we obtain $\lambda'=\lambda$, and so $q_j'=q_j$ for $j=1,\ldots,n$. 	
\end{proof}

The condition of $\pi_0(\Pi)$ being a tight operator frame in $\mathcal L_s^0(\mathbb C^d)$ brings to mind the definition of a tight IC-POVM given by Scott in \cite{Scott06}.
Let us introduce the following notation: $P_j:=\Pi_j/\tr\Pi_j$, $P:=(P_j)_{j=1}^n$ and
$\mathcal F:=\sum_{j=1}^n\tr\Pi_j|P_j)( P_j|=\sum_{j=1}^n\frac{1}{\tr\Pi_j}|\Pi_j)(\Pi_j|$.
\begin{Df} A POVM
$\Pi$ is a \textit{tight IC-POVM} if $\pi_0(P)$ is a tight operator frame with respect to the trace measure $\tau$ ($\tau(j)=\tr\Pi_j$, $j=1,\ldots,n$)
in $\mathcal L_s^0(\mathbb C^d)$, i.e. if there exists $\beta>0$ such that
$$\sum_{j=1}^n \tr\Pi_j|\pi_0(P_j))(\pi_0(P_j)|=\beta\mathcal I_0.$$
\end{Df}
\begin{R} The definition of a tight IC-POVM can be expressed equivalently in the following way:
$$\mathcal F=\beta\mathcal I+\frac{1-\beta}{d}|I)(I|,$$ where $\mathcal I$ stands for the identity superoperator on $\mathcal L_s(\mathbb C^d)$.
Note that $\beta$ can be calculated by taking the trace on both sides of one of the above equations:
$$\beta=\frac{1}{d^2-1}\left(\sum_{j=1}^n\frac{\tr\Pi_j^2}{\tr\Pi_j}-1\right).$$
\end{R}
Applying Theorem \ref{thmsim}, we get the following statement as a simple observation.
\begin{Cor}
\label{Scott}
 Let $\Pi$ be a POVM consisting of effects of equal trace.
Then $\Delta_\Pi$ is similar to the quantum state space if and only if $\Pi$ is a tight IC-POVM.
The square of the similarity ratio $\alpha$ is then equal to $\frac{1}{d^2-1}\left(\sum_{j=1}^n\tr\Pi_j^2-\frac{d}{n}\right)$.
\end{Cor}
Observe, however, that the equal-trace assumption cannot be dropped.
\begin{Ex}[A morphophoric POVM which is not a tight IC-POVM]
\label{Scott_example}
	Let $X$, $Y$ and $Z$ denote the Pauli matrices and let  $a:=\frac{\sqrt 3-1}{4}$ and $b:=\frac{3-\sqrt{3}}{6}$. Put $$\Pi_1:=aZ+aI,\quad \Pi_2:=-aZ+aI,\quad \Pi_3:=bX+bI,\quad \Pi_4:=-\frac{b}{2}X+aY+bI,\quad \Pi_5:=-\frac{b}{2}X-aY+bI.$$
	Then $(\Pi_j)_{j=1}^5$ is a rank-1 POVM which is morphophoric but not tight. Indeed, since  $\{X,Y, Z,I\}$ forms an (equal-norm) orthogonal basis of $\mathcal L_s(\mathbb C^2)$, it is easy to see that the projections of $\Pi_j$, $j=1,\ldots,5$, onto the 3-dimensional subspace of traceless operators can be identified as vertices of a triangular bipyramid: $(0,0,a)$, $(0,0,-a)$, $(b,0,0)$, $(-b/2,a,0)$, $(-b/2,-a,0)$, constituting a tight frame, see Fig. \ref{exnottight}. On the other hand, non-tightness of IC-POVM can be easily verified by direct computation.
	\begin{figure}[htb]
	\begin{center}	\includegraphics[scale=0.6]{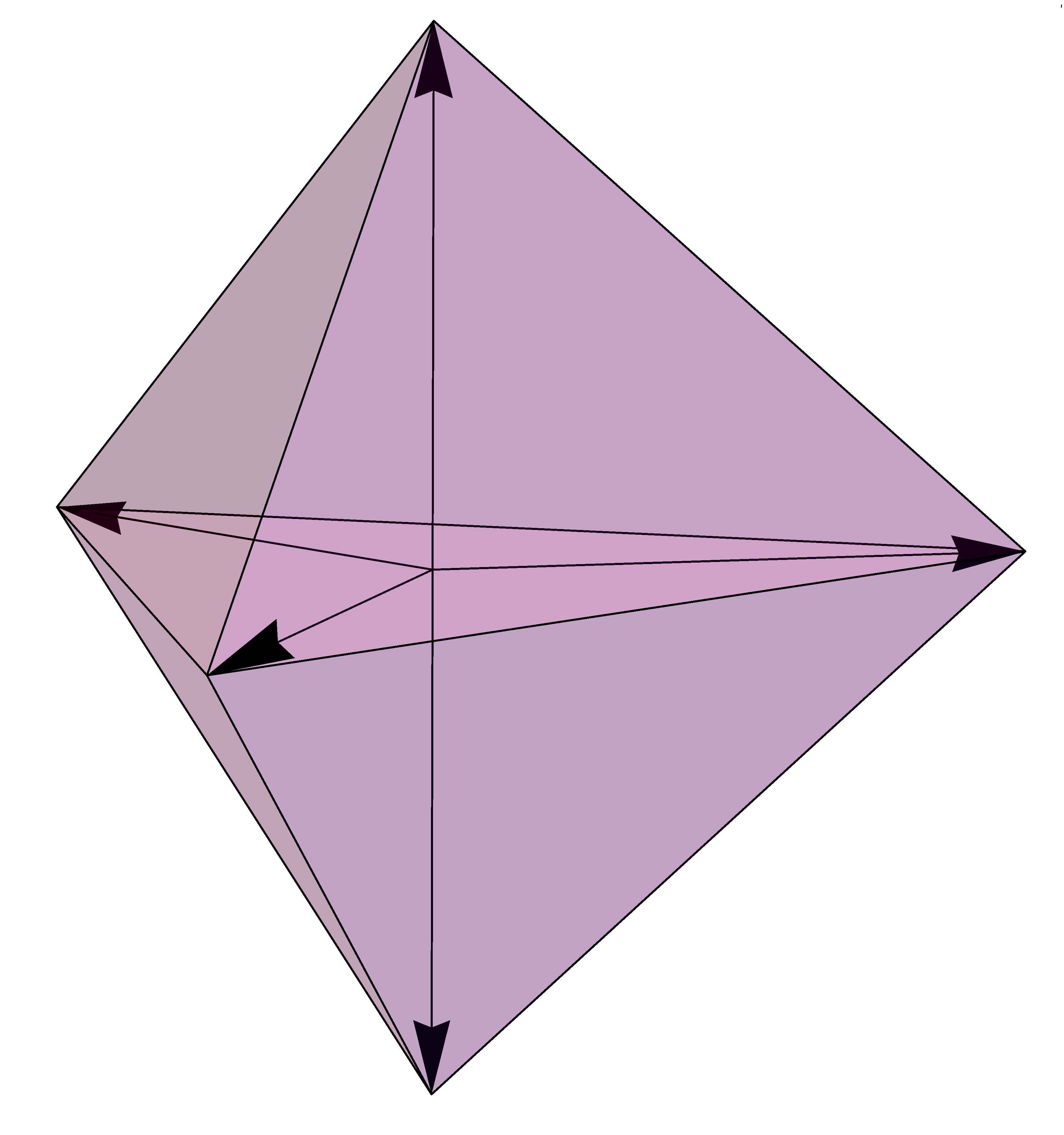}
		\end{center}
		\caption{A triangular bipyramid with vertices at $(0,0,a)$, $(0,0,-a)$, $(b,0,0)$, $(-b/2,a,0)$, $(-b/2,-a,0)$. The set of vertices forms a tight frame and therefore represents projections of effects constituting a morphophoric POVM.}\label{exnottight}
	\end{figure}	
\end{Ex}
\section{The probability range of 2-design POVMs}
\label{The probability range of 2-design POVMs}

\begin{Df}
We say that $\{\rho_j\}_{j=1}^n\subset \mathcal P(\mathbb C^d)$ is a $t$-design if
$$\frac{1}{n^2}\sum_{j,k=1}^nf(\tr(\rho_j\rho_k))=\iint_{(\mathcal P(\mathbb C^d))^2}f(\tr(\rho\sigma))d\mu(\rho)d\mu(\sigma)$$
holds for every $f:\mathbb R\to\mathbb R$  polynomial of degree $t$ or less, where by $\mu$ we mean the unique unitarily invariant measure on $\mathcal P(\mathbb C^d)$.
\end{Df}
Obviously, any $t$-design is also an $s$-design for $s<t$. Note that $\{\rho_j\}_{j=1}^n\subset \mathcal P(\mathbb C^d)$ is a 1\nobreakdash-design if and only if $(\frac{d}{n}\rho_j)_{j=1}^n$ is a POVM. Moreover, it is a 2-design if and only if $(\frac{d}{n}\rho_j)_{j=1}^n$ is a tight IC\nobreakdash-POVM \cite[Prop.~13]{Scott06}. We will refer to such rank-1 POVM as a 2-design POVM. That property allows us to rephrase the definition of a 2-design in a more comprehensible way, namely $\{\rho_j\}_{j=1}^n\subset \mathcal P(\mathbb C^d)$ is a~2-design if and only if
\begin{equation}\label{2des}
\tau=(d+1) \sum_{j=1}^n\frac{d}{n}\tr(\tau\rho_j)\rho_j-I
\end{equation}
for every $\tau\in\mathcal L_s(\mathbb C^d)$ such that $\tr\tau=1$.

\begin{Cor}\label{cor2des} Let $\Pi$ be a rank-1 POVM consisting of effects of equal trace. Then
$\mathcal {Q}_\Pi$ is similar to the quantum state space if and only if $\Pi$ is a 2-design POVM.
The square of the similarity ratio $\alpha$ is then equal to $\frac{d}{n(d+1)}$.
\end{Cor}

From now on we assume that $\Pi=(\frac{d}{n}\rho_j)_{j=1}^n$ is a 2-design POVM.
For $\tau\in\mixed$, $p_j(\tau):=(p_\Pi(\tau))_j=\frac{d}{n}\tr(\tau\rho_j)$ is the
probability of obtaining the $j$-th outcome. The following equalities are derived directly from \eqref{2des}:
\begin{align}
\sum_{j=1}^np_j(\tau)&=1,\\
\label{sqaresgeneral}
\sum_{j=1}^n(p_j(\tau))^2&=\frac{d (\tr(\tau^2)+1)}{n (d+1)},\\
\label{triplegeneral}
\sum_{j,k,l=1}^np_j(\tau)p_k(\tau)p_l(\tau)\tr(\rho_j\rho_k\rho_l)&= \frac{\tr(\tau+I)^3}{(d+1)^3}.
\end{align}
In particular, for pure states we get
\begin{align}\label{squares}
\sum_{j=1}^n(p_j(\tau))^2&=\frac{2d}{n (d+1)},\\
\label{triple}
\sum_{j,k,l=1}^np_j(\tau)p_k(\tau)p_l(\tau)\tr(\rho_j\rho_k\rho_l)&= \frac{d+7}{(d+1)^3}.
\end{align}

\begin{R}
Note that the rhs of the last equation does not depend on $n$.
\end{R}
The proof of the following statement can be found in \cite{JonLin05}.
\begin{F}
 A self-adjoint operator $\tau$ is a pure quantum state if and only if $\tr\tau=1$, $\tr\tau^2=1$ and $\tr\tau^3=1$.
\end{F}
This fact will be helpful in proving the next theorem that characterises algebraically the probability distributions belonging to $\mathcal {Q}_\Pi^1$.
In spite of its apparent similarity to \cite[Cor. 2.3.4.]{GraPhD}, the statement is more general and has different meaning
since we do not assume that the vector satisfying conditions \eqref{itm:lin}-\eqref{itm:cub} below is a probability distribution from $\mathcal {Q}_\Pi$.
In fact, we do not even assume that it is a probability distribution at all.
\begin{Th}\label{thm2des} Let $\Pi=(\frac{d}{n}\rho_j)_{j=1}^n$ be a 2-design POVM. Then $(p_1,\ldots,p_n)\in\mathcal {Q}_{\Pi}^1$ if and only if
\begin{enumerate}[i.]
	\item\label{itm:lin} $\frac{n}{d}p_l=(d+1)\sum_{j=1}^np_j\tr(\rho_j\rho_l)-1$ for $l=1,\ldots,n$;
	\item\label{itm:squ} $\sum_{j=1}^n p_j^2=\frac{2d}{n (d+1)}$;
	\item\label{itm:cub} $\sum_{j,k,l=1}^np_jp_kp_l\tr(\rho_j\rho_k\rho_l)=\frac{d+7}{(d+1)^3}$.
\end{enumerate}
\end{Th}

\begin{proof} We already know that every probability distribution from $\mathcal {Q}_{\Pi}^1$ satisfies all the equations above.
Thus, to complete the proof, it suffices to show that any vector $(p_1,\ldots,p_n)\in\mathbb R^n$ fulfilling equalities
\eqref{itm:lin}-\eqref{itm:cub} belongs to $\mathcal {Q}_\Pi^1$.   Let us define $\tau:=(d+1)\sum_{j=1}^np_j\rho_j-I$.
Then $$\tr(\tau\rho_l)= (d+1)\sum_{j=1}^np_j\tr(\rho_j\rho_l)-\tr(\rho_l)=\frac{n}{d}p_l$$
and thus $p_l=\frac{d}{n}\tr(\tau\rho_l)$. It is now enough to prove that $\tau\in\mathcal P(\mathbb C^d)$. Obviously, $\tau^*=\tau$.
Using \eqref{itm:lin}, we show that the numbers $p_1,p_2,\ldots,p_n$ sum up to 1:
$$\sum_{l=1}^n p_l=\frac{d}{n}\left( (d+1)\sum_{j=1}^np_j\tr\left(\rho_j\sum_{l=1}^n\rho_l\right)-n\right)=\frac{d}{n} (d+1) \frac{n}{d} \sum_{j=1}^np_j-\frac{d}{n} n=(d+1)\sum_{j=1}^np_j-d,$$
which gives $\sum_{l=1}^n p_l=1$. In consequence, $\tr\tau=(d+1) \sum_{j=1}^np_j\tr(\rho_j)-d=1$. Since \eqref{sqaresgeneral}
and \eqref{triplegeneral} hold true for any Hermitian operator with trace equal to 1, we can compare them with our assumptions
concerning the sum of squares and triple products, respectively, to obtain $\tr\tau^2=\tr\tau^3=1$, which completes the proof.
\end{proof}

Let us take a closer look at equations \eqref{itm:lin}-\eqref{itm:cub} from Theorem \ref{thm2des}. As we already know from Corollary~\ref{cor2des}, they need to  describe a $(2d-2)$-dimensional submanifold of the $(d^2-2)$-dimensional sphere. From the proof above we deduce that  $p\in\mathbb R^n$ fulfils the system of $n$ linear equations \eqref{itm:lin} if and only if $p$ lies in the affine span of $\mathcal Q_\Pi$.
\begin{Df}\label{A} By the \emph{primal affine space} we mean the affine span of $\mathcal Q_\Pi$ and denote it by $\mathcal A$.
\end{Df}
 Note that $\mathcal A\subset\{x\in\mathbb R^n:x_1+\ldots+x_n=1\}$ and $\dim\mathcal A=d^2-1$, so the number of equations can be reduced to $n-(d^2-1)$. In particular, if $n=d^2$, this system of equations reduces to the single equation $\sum_{j=1}^{d^2} p_j=1$.  The only 2-design POVMs with $d^2$ elements are  SIC-POVMs, i.e. those with $\tr(\rho_j\rho_k)=\frac{1}{d+1}$ for $j\neq k$.  The quadratic form \eqref{itm:squ} is obviously the equation of the sphere in $\mathbb R^n$ centred at 0 and of radius $\sqrt{\frac{2d}{n(d+1)}}$. The intersection of this sphere with the affine subspace $\mathcal A$ gives us the $(d^2-2)$-dimensional sphere we are looking for. Finally, the cubic form \eqref{itm:cub} is responsible for cutting the $(2d-2)$-dimensional submanifold from the sphere in question. We can now state the following:
\begin{Cor}\hfill
\label{Special}
\begin{enumerate}[i.]
\item The set $\mathcal {Q}_\Pi^1$ is fully described by  equations \eqref{itm:lin}-\eqref{itm:squ} if and only if $d=2$.
\item The linear equations \eqref{itm:lin} can be reduced to the single equation $\sum_{j=1}^n p_j=1$ if and only if $\Pi$ is a~SIC-POVM.
\end{enumerate}
\end{Cor}
A natural question arises whether  we can provide a similar characterisation of $\mathcal {Q}_\Pi$ as we did  for $\mathcal {Q}_\Pi^1$ in Theorem \ref{thm2des}. As we can see, the proof of the latter is based on the fact that the pure states are fully described by $\tr\tau=\tr\tau^2=\tr\tau^3=1$. The mixed states obviously satisfy $\tr\tau=1$ and $\tr\tau^2\leq 1$ and it is tempting to write down the third inequality in a similar manner. However, it turns out that in order to provide a complete characterisation of $\tau\in\mathcal S(\mathbb C^d)$ in terms of $\tr\tau^k$, we need additional $d-2$ inequalities. Put together, they take the form $S_k\geq 0$ ($k\in\{2,\ldots,d\}$), where $S_k$ can be defined recursively by $S_k=\frac{1}{k}\sum_{j=1}^k(-1)^{j-1}\tr(\tau^j)S_{k-j}$ and $S_0=S_1=1$ \cite{ByrKha03}. For example, the condition on $\tr\tau^3$ is as follows: $3\tr\tau^2-2\tr\tau^3\leq1$. Thus, such a characterisation is possible but it becomes more and more complicated with the growth of $d$. However, for $d=2,3$ we can characterise $\mathcal Q_\Pi$ in a relatively simple way.
\begin{Cor}\label{ineq2des} Let $\Pi=(\frac{d}{n}\rho_j)_{j=1}^n$ be a 2-design POVM. If $(p_1,\ldots,p_n)\in\mathcal {Q}_{\Pi}$, then
\begin{enumerate}[i.]
	\item $\frac{n}{d}p_l=(d+1)\sum_{j=1}^np_j\tr(\rho_j\rho_l)-1$ for $l=1,\ldots,n$;
	\item \label{sqineq}$\sum_{j=1}^n p_j^2\leq\frac{2d}{n (d+1)}$;
	\item $\sum_{j,k,l=1}^np_jp_kp_l\tr(\rho_j\rho_k\rho_l)\geq\frac{9n}{2d(d+1)^2}\sum_{j=1}^n p_j^2+\frac{d-2}{(d+1)^3}$.
	\end{enumerate}
	Moreover, conditions (i)-(iii) are sufficient for $d=3$ and (i)-(ii) are sufficient for $d=2$.
\end{Cor}

\section{Examples}
\label{Examples}

In the following examples we provide some explicit formulae for the $(d^2-1)$-dimensional affine subspaces $\mathcal A$ defined by the system of linear equations \eqref{itm:lin} in Theorem \ref{thm2des}.
\begin{Ex}[Cubical POVM]
Let us consider a cubical POVM $\Pi=(\frac{1}{4}\rho_j)_{j=1}^8$, i.e. the states $\rho_1,\ldots,\rho_8$ are represented on the Bloch sphere by the vertices of a cube, e.g. $\frac{1}{\sqrt{6}}(\pm1,\pm1,\pm1)$. Let us label these vertices in the following way: the vertices of the bottom face are labelled in sequence  with 1-4 and if $i$~and $j$ are the labels of the opposite vertices, then $j=i+4$ ($mod\ 8$). The eight linear equations \eqref{itm:lin} from Theorem \ref{thm2des} reduce to the following $8-(2^2-1)=5$ linearly independent ones:
$$\begin{cases}
p_1+p_5=\frac{1}{4}\\
p_2+p_6=\frac{1}{4}\\
p_3+p_7=\frac{1}{4}\\
p_4+p_8=\frac{1}{4}\\
p_1+p_3+p_6+p_8=\frac{1}{2}.
\end{cases}$$
It is easy to see that the first four of them guarantee that $p_1+\ldots+p_8=1$. Obviously, the 3-dimensional affine subspace of $\mathbb R^8$ defined by these equations is label-dependent. The number of possible labellings is $8!=40320$. However, two of them produce the same set of equations, if the  permutation transforming one to the other corresponds to an isometric operation on the cube. The isometries of the cube are described by the elements of the octahedral group $O_h$. Thus the number of different subspaces is $\frac{8!}{|O_h|}=\frac{40320}{48}=840$.
\end{Ex}

\begin{Ex}[Cuboctahedral and icosahedral POVMs] The polyhedra  defining these POVMs (cuboctahedron and icosahedron) have the same number of vertices: $n=12$, which can be represented on the Bloch sphere by $\frac{1}{\sqrt{2(s^2+1)}}(\pm s,\pm 1,0)$, $\frac{1}{\sqrt{2(s^2+1)}}(0,\pm s,\pm 1)$ and $\frac{1}{\sqrt{2(s^2+1)}}(\pm 1,0,\pm s)$, where $s=1$ for cuboctahedron and $s=(1+\sqrt{5})/2$ (\emph{golden ratio}, usually denoted by $\varphi$) for icosahedron. If we label these vertices as follows: $v_1=(+,+,0)$, $v_3=(+,-,0)$, $v_5=(0,+,+)$, $v_7=(0,+,-)$, $v_9=(+,0,+)$, $v_{11}=(-,0,+)$ and $v_{2j}=-v_{2j-1}$ for $j\in\{1,\ldots,6\}$, then the different values of the parameter $s$ result in slightly  different affine subspaces defined by the following systems of 9 linear equations:
$$\textnormal{cuboctahedron:} \begin{cases}
p_1+p_2=\frac{1}{6}\\
p_3+p_4=\frac{1}{6}\\
p_5+p_6=\frac16\\
p_7+p_8=\frac16\\
p_9+p_{10}=\frac16\\
p_{11}+p_{12}=\frac16\\
p_1-p_3+p_6+p_8=\frac16\\
p_5-p_7+p_{10}+p_{12}=\frac16\\
p_9-p_{11}+p_2+p_4=\frac16
\end{cases}, \quad\textnormal{icosahedron:}\begin{cases}
p_1+p_2=\frac{1}{6}\\
p_3+p_4=\frac{1}{6}\\
p_5+p_6=\frac16\\
p_7+p_8=\frac16\\
p_9+p_{10}=\frac16\\
p_{11}+p_{12}=\frac16\\
\varphi(p_1-p_3)+p_6+p_8=\frac16\\
\varphi(p_5-p_7)+p_{10}+p_{12}=\frac16\\
\varphi(p_9-p_{11})+p_2+p_4=\frac16
\end{cases}.$$
As in the previous case, the definitions of the subspaces are label-dependent. The number of possible labellings is $12!$, but taking into account the symmetry groups of cuboctahedron and icosahedron, i.e. the octahedral group $O_h$ and the icosahedral group $I_h$, we get the numbers of different subspaces are  $\frac{12!}{|O_h|}=\frac{12!}{48}=9979200$ and $\frac{12!}{|I_h|}=\frac{12!}{120}=3991680$.
\end{Ex}

Starting from dimension $d=3$, the equation of third degree plays a crucial role in defining $\mathcal {Q}_\Pi^1$. In the following example we take a closer look at the 1-parameter family of 2-design POVMs for which the equations of first and second degree coincide, but not these of third degree. But first let us observe that the lhs of \eqref{itm:cub} in Theorem \ref{thm2des} can be written in a slightly different way if we  assume that both \eqref{itm:lin} and \eqref{itm:squ} from the same theorem hold. We start with:
$$\sum_{j,k,l=1}^np_jp_kp_l\tr(\rho_j\rho_k\rho_l)= \sum_{j=1}^np_j^3\tr(\rho_j^3)+3\sum_{j\neq k}p_j^2p_k\tr(\rho_j^2\rho_k)+\sum_{j\neq k\neq l\neq j}p_jp_kp_l\tr(\rho_j\rho_k\rho_l).$$
But $\rho_j^2=\rho_j$ and $\tr(\rho_j^3)=1$. Moreover,
\begin{align*}
\sum_{j\neq k}p_j^2p_k\tr(\rho_j\rho_k)&= \sum_{j=1}^np_j^2\left(\sum_{k=1}^np_k\tr(\rho_j\rho_k)-p_j\tr(\rho_j^2)\right)=\sum_{j=1}^np_j^2\left(\frac{1}{d+1}\left(\frac{n}{d}p_j+1\right)-p_j\right)\\
&=\left(\frac{n}{d(d+1)}-1\right)\sum_{j=1}^np_j^3+\frac{1}{d+1}\sum_{j=1}^np_j^2=\left(\frac{n}{d(d+1)}-1\right)\sum_{j=1}^np_j^3+\frac{2d}{n(d+1)^2}.
\end{align*}
Thus, the third degree equation now takes the form:
\begin{equation}\label{new3rd}
\left(\frac{3n}{d(d+1)}-2\right)\sum_{j=1}^np_j^3+\sum_{j\neq k\neq l\neq j}p_jp_kp_l\tr(\rho_j\rho_k\rho_l)=\frac{n(d+7)-6d(d+1)}{n(d+1)^3}
\end{equation}
\begin{Ex}[SIC-POVMs in dimension 3]
 Let $$v_{0,j}^t=\frac{1}{\sqrt{2}}(-e^{it}\eta^j,0,1),\quad v_{1,j}^t=\frac{1}{\sqrt{2}}(1,-e^{it}\eta^j,0),\quad v_{2,j}^t=\frac{1}{\sqrt{2}}(0,1,-e^{it}\eta^j),$$ where $\eta=e^{2\pi i/3}$, $j\in\{0,1,2\}$ and $t\in[0,\pi/3)$. Then  $\Pi^t:=(\frac{1}{3}\rho_{m,j}^t)_{m,j=0}^2:=(\frac{1}{3}|v_{m,j}^t\rangle\langle v_{m,j}^t|)_{m,j=0}^2$ is a~SIC-POVM and for $t\neq t'$ the sets $\Pi^t$ and $\Pi^{t'}$ are not unitarily equivalent. It is easy to see that the equations of first and second degree are the same for every $t$, but that is not the case for the third degree equation. Let us introduce the following notation:
 \begin{align*}
J&:=\{(\alpha,\beta,\gamma)\in (\mathbb Z_3\times \mathbb Z_3)^3| \alpha\neq\beta\neq\gamma\neq\alpha\}  \\
J_k&:=\{(\alpha,\beta,\gamma)\in (\mathbb Z_3\times \mathbb Z_3)^3| \alpha_1\neq\beta_1\neq\gamma_1\neq\alpha_1, \alpha_2+\beta_2+\gamma_2=k\ (\textnormal{mod }3)\},\quad k=0,1,2 \\
J_3&:=\{(\alpha,\beta,\gamma)\in (\mathbb Z_3\times \mathbb Z_3)^3| \alpha_1=\beta_1=\gamma_1, \alpha_2\neq\beta_2\neq\gamma_2\neq\alpha_2\}\\
J'&:=J\setminus(J_0\cup J_1\cup J_2\cup J_3).
 \end{align*}
Now, using \eqref{new3rd}, we get (for the sake of clarity we omit the label $t$)
\begin{align*}
\frac{1}{32}&=\frac{1}{4}\sum_{\alpha\in\mathbb Z_3\times\mathbb Z_3}p_\alpha^3 +\sum_{(\alpha,\beta,\gamma)\in J}p_\alpha p_\beta p_\gamma \tr(\rho_\alpha\rho_\beta\rho_\gamma)\\
&=\frac{1}{4}\sum_{\alpha\in\mathbb Z_3\times\mathbb Z_3}p_\alpha^3-\frac{1}{8}\cos(3t)\sum_{(\alpha,\beta,\gamma)\in J_0}p_\alpha p_\beta p_\gamma
+\frac{1}{16}(\cos(3t)+\sqrt{3}\sin(3t))\sum_{(\alpha,\beta,\gamma)\in J_1}p_\alpha p_\beta p_\gamma\\&+\frac{1}{16}(\cos(3t)-\sqrt{3}\sin(3t))\sum_{(\alpha,\beta,\gamma)\in J_2}p_\alpha p_\beta p_\gamma-\frac{1}{8}\sum_{(\alpha,\beta,\gamma)\in J_3}p_\alpha p_\beta p_\gamma+\frac{1}{16}\sum_{(\alpha,\beta,\gamma)\in J'}p_\alpha p_\beta p_\gamma
\end{align*}
For $t=0$ we get the Hesse configuration and the above equation takes a particularly nice form \cite{TabApp}:
$$\sum_{\alpha\in\mathbb Z_3\times\mathbb Z_3}p_\alpha^3-\frac{1}{2}\sum_{(\alpha,\beta,\gamma)\in J_0\cup J_3}p_\alpha p_\beta p_\gamma=0.$$
\end{Ex}

\section{MUB-like POVMs}
\label{MUB-like POVMs}

\subsection{MUBs}
\label{MUBs}

Another remarkable example of a $2$-design POVM is a \textit{complete set of
mutually unbiased bases} (\textit{MUBs}) in $\mathbb{C}^{d}$, where $d$ is a
prime power. (The question whether complete sets of MUBs exist in other
dimensions remains open). Such POVM $\Pi$ consists of $d(d+1)$ effects of the
form $\frac{1}{d+1}\rho_{j}$, where $\rho_{j}\in\mathcal{P}\left(
\mathbb{C}^{d}\right)  $ for $j=1,2,\ldots,d(d+1)$, and
\[
\operatorname{tr}(\rho_{kd+l}\rho_{k^{\prime}d+l^{\prime}})=\left\{
\begin{tabular}
[c]{ll}
$1$ & $\text{\text{for }}k=k^{\prime},l=l^{\prime}$\\
$0$ & $\text{\text{for }}k=k^{\prime},l\neq l^{\prime}$\\
$\frac{1}{d}$ & for\text{ }$k\neq k^{\prime}$
\end{tabular}
\right.
\]
for $k,k^{\prime}\in\{0,\ldots,d\}$ and $l,l^{\prime}\in\{1,\ldots,d\}$.

The fact that a complete set of MUBs constitutes a $2$-design POVM was first shown
explicitly in \cite{KlaRot05}, however, this also follows easily \cite{BelSmo08} from
the fact that in this case $\Pi$\ fulfils both the \textit{Welch bound}
\cite{Wel74}\ and the \textit{Levenshtein bound} \cite{Lev92} for the angles
between $2$-design lines.

Now, the system of linear equations defining the primal affine space
$\mathcal{A}$ takes a particularly nice form since, instead of the initial
$n=d(d+1)$ equations, it suffices to consider the following $d+1$ equations
expressing the simple fact that the probabilities over any basis ($k=0,\ldots
,d$) need to sum up to $\frac{1}{d+1}$:
\[
\sum_{l=1}^{d}p_{kd+l}=\frac{1}{d+1}\text{.}
\]
We shall show that this result is valid  even in more general settings, see
Theorem \ref{Bases} below.

\subsection{MUB-like 2-designs: definition and the orthogonality graph}
\label{MUB-like 2-designs: definition and the orthogonality graph}

Let us consider a $2$-design POVM $\Pi=( \frac{d}{n}\rho_{j})
_{j=1}^{n}$, where $\rho_{j}\in\mathcal{P}\left(  \mathbb{C}^{d}\right)  $ for
$j=1,\ldots,n$. We call $\Pi$ \textit{MUB-like} \cite{HugWal}\ if there
exists $a>0$ such that $\left\{  \operatorname{tr}\left(  \rho_{j}\rho
_{k}\right)  :j,k=1,\ldots,n,j\neq k\right\}  =\left\{  0,a\right\}  $. The
\textit{orthogonality graph} for $\Pi$ is the graph $\Gamma$ with the vertex
set $\left\{  \rho_{j}\right\}  _{j=1}^{n}$ and the adjacency relation defined
by $\rho\sim\sigma$ iff $\rho\perp\sigma$ iff $\operatorname{tr}\left(
\rho\sigma\right)  =0$. Now, we can apply the results of Hoggar, who, in even
more general setting, examined the properties of $\Gamma$ \cite[Theorem
5.2]{Hog82}\cite[Sec. 3]{Hog84}, employing the unpublished notes of Neumaier
\cite{Neu81}. In particular, he showed that $\Gamma$ is a distance regular
graph with diameter $2$, and so strongly regular, see \cite{Broetal89}. We now
specify Hoggar's results for MUB-like $2$-design POVMs. (Note, however, that
Hoggar analysed the graph complementary to $\Gamma$.) Denote by:

\begin{itemize}
\item $\varkappa$ - the number of vertices of $\Gamma$ adjacent (orthogonal)
to a given vertex;

\item $\lambda$ (resp. $\mu$) - the number of vertices of $\Gamma$ adjacent to
a pair of adjacent (resp. not adjacent) vertices.
\end{itemize}

Obviously, $\Gamma$ is non-complete and connected unless $\mu=0$. The
adjacency matrix of $\Gamma$\ has three eigenvalues $\varkappa$, $r$, and
$-q$, where $r,q\in\mathbb{N}$, $1\leq q\leq r^{2}$, $\psi:=\frac{r(q-1)}
{r+q}\in\mathbb{N}$, with the multiplicities $1$, $f$, and $g$, respectively.
Note that the inequality $q\leq r^{2}$ follows from the
Delsarte-Goethals-Seidel bound \cite{Deletal75}, see also \cite[Theorem
7.2]{LehTay09}, for the cardinality of sets of lines in $\mathbb{C}^{d}$
having a prescribed number of angles. Moreover, $a=\frac{1}{r+1}$, and
\begin{align*}
d  &  =r+\psi+1\text{, \ \ }n=d(r+2q)\text{,}\\
\varkappa &  =(d-1)q\text{, \ \ }\mu=\psi q\text{, \ \ }\lambda=\mu
+r-q\text{,}\\
f  &  =n-d^{2}\text{, \ \ }g=d^{2}-1\text{.}
\end{align*}

Clearly, $\Pi$\ saturates both the Welch bound and the Levenshtein bound,
which are in this case equivalent as%
\begin{equation}
a=\frac{2n-d(d+1)}{(n-d)(d+1)}\text{.}\label{WelLev}%
\end{equation}
In fact, as observed in \cite[Proposition 3]{Ega19},
if $\Pi$ is a rank-one POVM consisting of effects of equal trace such that
$\left\{0,a\right\} = \left\{  \operatorname{tr}\left(  \rho_{j}\rho_{k}\right)
:j,k=1,\ldots,n,j\neq k\right\}$, then $\Pi$ forms a
$2$-design if and only if \eqref{WelLev} holds.

Clearly, a maximal clique (full subgraph) in $\Gamma$, i.e. a maximal
orthogonal set, has at most $d$ elements, which is just the Hoffman-Delsarte
bound for $\Gamma$ \cite{Soi15}, as $1+\varkappa/q=d$. Let us recall that a
clique $C$ is called \textit{regular} if every point not in $C$ is adjacent to
the same number (the \textit{nexus}) of points in $C$. It follows from
\cite[Theorem 1.1, Corollary 1.2]{Neu81b} that all regular cliques in $\Gamma
$, if exist, have $d$ elements (and so they are so called \textit{Delsarte
cliques}) with the nexus equal to $\psi$. From \cite[Theorem 3.1.iv]{Soi15} we
deduce that every clique in $\Gamma$ of size $d$ is regular. Hence, $\Gamma$
is a~\textit{Delsarte clique graph }(for definition see \cite{Banetal07}) if
and only if for every $j,k=1,\ldots,n$, $j\neq k$ such that $\rho_{j}\perp
\rho_{k}$, $\left\{  \rho_{j},\rho_{k}\right\}  $ is contained in the same
number (say $c$) of bases. More precisely, $\Gamma$ is then a~Delsarte clique
graph with parameters $\left(  \varkappa,d-1,c\right)  $. Moreover, double
counting argument shows that the number of maximal cliques in every local
graph is $cq$, and the number of maximal cliques in $\Gamma$ is $\left(
r+2q\right)  cq$. Note that if $\Gamma$ is an \textit{edge transitive graph}
and there is a~Delsarte clique in $\Gamma$, then $\Gamma$ is a Delsarte clique
graph. In particular, this is true for \textit{rank }$3$\textit{ graphs}, i.e.
strongly regular graphs generated by a \textit{rank }$3$\textit{ permutation
group} of even order acting on a finite set $X$ (i.e. a group with exactly
three orbits on $X\times X$), as they are edge transitive \cite{Hig64}.

\subsection{Examples of MUB-like 2-designs}
\label{Examples2}

Let us now consider two extreme cases: $q=1$ and $q=r^{2}$, where $n$ is,
resp., minimal and maximal, for a given $d$.

\begin{itemize}
\item $q=1$ (iff $\psi=0$) \smallskip

In this case $a=\frac{1}{d}$, $n=d(d+1)$, $\varkappa=d-1$, $\lambda=d-2$ and
$\mu=0$. In consequence, $\Pi$ must be a \textit{complete set of mutually
unbiased bases} (\textit{MUBs}).

\item $q=r^{2}$ (iff $\psi=r-1$) \smallskip

Then $d=2r$, $a=\frac{1}{r+1}$, $n=2r^{2}\left(  2r+1\right)  =d^{2}\left(
d+1\right)  /2$, $\varkappa=\left(  2r-1\right)  r^{2}$, $\lambda=\left(
r-1\right)  ^{2}r$ and $\mu=\left(  r-1\right)  r^{2}$. Hoggar showed
\cite[Theorem 3.7]{Hog84} that this is also a sufficient and necessary
condition for $\Pi$\ to be a (minimal or tight) $3$\textit{-design}. Note that
$\Pi$ cannot be a $4$-design \cite[Theorem 5.2]{Hog82}.
\end{itemize}

Note that the case $q=r=1$, i.e. the POVM consisting of $3$ MUBs in dimension
$d=2$ that form the regular octahedron inscribed in the Bloch sphere,\ lies at
the intersection of these two extremes. Besides complete sets of MUBs, which
exist for every prime power $d$, but not necessarily for other dimensions, we
know four other MUB-like POVMs, one in each of dimensions $4$, $5$, $6$ and $28$,
see also \cite{Ega19}. The parameters of these POVMs and their orthogonality
graphs are collected in Tab. \ref{Tab:Graphs}.

\begin{table}[!htbp]
	\begin{center}
\begin{tabular}
[c]{|c|c|c|c|c|c|c|c|c|}\hline
$n$ & $\varkappa$ & $\lambda$ & $\mu$ & $d$ & $r$ & $q$ & $\psi$ & $c$\\\hline
$d(d+1)$ & $d-1$ & $d-2$ & $0$ & $d$ & $d-1$ & $1$ & $0$ & $1$\\\hline
$40$ & $12$ & $2$ & $4$ & $4$ & $2$ & $4$ & $1$ & $1$\\\hline
$45$ & $12$ & $3$ & $3$ & $5$ & $3$ & $3$ & $1$ & $1$\\\hline
$126$ & $45$ & $12$ & $18$ & $6$ & $3$ & $9$ & $2$ & $3$\\\hline
$4060$ & $1755$ & $730$ & $780$ & $28$ & $15$ & $65$ & $12$ & $45$\\\hline
\end{tabular}
\end{center}
\caption{Parameters of the orthogonality graphs generated by MUB-like POVMs.}
\label{Tab:Graphs}
\end{table}
Note that all these graphs are Delsarte clique graphs:
\begin{itemize}
\item  For MUBs $\Gamma$ is the disjoint union of maximal cliques and, in
consequence, a Delsarte clique graph with $\psi=0$ and $c=1$.\smallskip

\item $d=4$ \cite[Ex. 7]{Hog82} This $3$-design was introduced in \cite[Sec. 6]{Pen92} under the name of \textit{Penrose dodecahedron}, see
also \cite{ZimPen93,Pen94,MasAra99}, and, independently,
  as the diameters of the
\textit{Witting polytope} \cite{Cox74,CoxShe92}. Waegell and Aravind
\cite{WaeAra17} showed that these two constructions are in fact equivalent.
Hoggar \cite{Hog83} proved that in this case $\Gamma$ is the collinearity
graph of one of two generalised \textit{quadrangles of order }$(3,3)$
\cite{Pay75}\cite[Secs. 3.1.1 \& 6.2.1]{PayTha09}, namely $W(3)$ arising as the
set of absolute points and lines of a symplectic polarity of projective
geometry $PG(3,3)$, see also \cite[Ex. 10.10]{BroMal}. Hence $\Gamma$ is a
(rank $3$) Delsarte clique graph with $\psi=1$ and $c=1$.\smallskip

\item $d=5$ \cite[Ex. 18]{Hog82} This $2$-design was defined in
\cite{Deletal75} and it is equal to the \textit{line system }$\mathcal{K}%
_{5}$ \cite[pp. 107 \& 166]{LehTay09}. We shall show below (Ex. \ref{dim5}) that
in this case $\Gamma$ is the collinearity graph of the unique \textit{generalised
quadrangle of order }$(4,2)$ \cite{Hir94}\cite[Ex. 10.11]{BroMal}. In consequence,
$\Gamma$ is a (rank $3$) Delsarte clique graph with $\psi=1$ and $c=1$.\smallskip

\item $d=6$ \cite[Ex. 7]{Hog82} This $3$-design was introduced in
\cite{Mit1914} and it is just the \textit{line system }$\mathcal{K}_{6}$
\cite[pp.~107 \& 166]{LehTay09}. The orthogonality graph for this design is a
unique rank 3 strongly regular graph with parameters $(126,45,12,18)$
\cite{BloBro84}\cite[Ex. 10.32]{BroMal}, having the collinearity graph of
unique quadrangle of order $(4,2)$, see above, as its local graph \cite[Sec.
5.2]{Crnetal10}. Hence, $\Gamma$ is a~(rank $3$) Delsarte clique graph with
$\psi=2$ and $c=3$.\smallskip

\item $d=28$ \cite[Ex. 20]{Hog82} This MUB-like $2$-design (with $a=1/16$) was
defined in \cite{ConWal73} and it comes from a $28$-dimensional representation
of the sporadic simple Rudvalis group. This group is a~rank $3$ permutation group
of even order \cite{Hig64}\ on a $4060$-element set and, as such, it generates an
edge transitive strongly regular (\textit{Rudvalis}) \textit{graph}, which is
the orthogonality graph for the $2$-design. Now, $\Gamma$ is a (rank $3$)
Delsarte clique graph with $\psi=12$ and $c=45$ \cite[Ex. 10.65]{BroMal}.\smallskip
\end{itemize}

However, unfortunately, we do not know neither whether it is true that the orthogonality graph of a~MUB-like POVM is always a Delsarte clique graph, nor even whether any other MUB-like POVMs apart from those discussed above exist. The answer to this question seems to be of some importance, as we use this property in the next section.

\subsection{Primal affine space in MUB-like case}
\label{Main result}

The next result gives us an alternative description of the primal affine space (see Definition \ref{A}), and ipso facto an alternative form
of the particular case of the primal equation  in the MUB-like case,
assuming that the orthogonality graph $\Gamma$ is a~Delsarte clique graph. On the other hand, it
shows that the primal polytope (see Section \ref{Geometric properties of generalised qplex})
$\mathcal{A\cap}\Delta_{n} = \frac{d}{n}\cdot\operatorname{QSTAB}(\Gamma)\mathcal{\cap}\Delta_{n}$,
where $\operatorname{QSTAB}(\Gamma)$ is the \textit{clique constrained stable
set polytope} for $\Gamma$ \cite{AmaCun18}.

\begin{Th}
\label{Bases}Let $\Pi$ be a MUB-like POVM and let its orthogonality graph be a
Delsarte clique graph. Denote by $\mathcal{C}$  the set of orthogonal bases
(maximal cliques) contained in $\left\{  \rho_{j}\right\}  _{j=1}^{n}$. For $C\in\mathcal{C}$ set
$J_{C}:=\left\{  j=1,\ldots,n:\rho_{j}\in C\right\}$.
Then the following conditions are equivalent for $p\in\left(\mathbb{R}^{+}\right)^{n}$:
\begin{enumerate}[i.]
\item $p$ fulfils (i) from Theorem \ref{thm2des};

\item $p\in\Delta_{n}$ and for every $C\in\mathcal{C}$%

\begin{equation}
\sum_{j\in J_{C}}p_{j}=\frac{d}{n}\text{;} \label{MUBur}%
\end{equation}

\item $p\in\Delta_{n}$ and for every clique $C$ in $\Gamma$%
\begin{equation}
\sum_{j\in J_{C}}p_{j}\leq\frac{d}{n}\text{.} \label{MUBin}%
\end{equation}
\end{enumerate}
\end{Th}

\begin{proof}
For $l=1,\ldots,n$ denote by
$\mathcal{C}_{l}:=\{C\in\mathcal{C:}\rho_{l}\in C\}$,
$A_{l}:=\{j=1,\ldots,n:\operatorname{tr}\left( \rho_{j}\rho_{l}\right) =a\}$ and
$O_{l}:=\{j=1,\ldots,n:\operatorname{tr}\left(  \rho_{j}\rho_{l}\right)  =0\}$.
Note that $\left| \mathcal{C}_{l}\right|=cq$.\smallskip

The implication $(ii)\Rightarrow(iii)$ follows from the fact that each clique is
contained in a maximal one.\smallskip

To show $(iii)\Rightarrow(i)$, set $l=1,\ldots,n$. Applying \eqref{MUBin},
we get
\[
\sum_{C\in\mathcal{C}_{l}}\sum_{j\in J_{C}}p_{j}\leq cq\frac{d}{n}\text{.}%
\]
In consequence,
\[
cqp_{l}+c\sum_{j\in O_{l}}p_{j}=cqp_{l}+\sum_{C\in\mathcal{C}_{l}}(\sum_{j\in
J_{C}}p_{j}-p_{l})\leq cq\frac{d}{n}\text{,}%
\]
and so
\begin{equation}
qp_{l}+\sum_{j\in O_{l}}p_{j}\leq q\frac{d}{n}\text{.}\label{in1}%
\end{equation}
Multiplying both sides of \eqref{in1}\ by $(d+1)a$, we get%
\begin{align*}
\frac{q}{r+q}=(d+1)aq\frac{d}{n}
& \geq (d+1)a(qp_{l}+\sum_{j\in O_{l}}p_{j})\\
& =(d+1)aqp_{l}+(d+1)a(1-p_{l}-\sum_{j\in A_{l}}p_{j})\\
& =(d+1)\left(  aq+(1-a)\right)  p_{l}-(d+1)(p_{l}+a\sum_{j\in A_{l}}
p_{j})+(d+1)a\\
& =\frac{n}{d}p_{l}-(d+1)(p_{l}+a\sum_{j\in A_{l}}p_{j})+\frac{r+2q}{r+q}.
\end{align*}
Thus,%
\begin{equation}
\frac{n}{d}p_{l}\leq(d+1)(p_{l}+a\sum_{j\in A_{l}}p_{j})-1\text{.}\label{in2}%
\end{equation}
On the other hand, the sums of both sides of \eqref{in2} over $l=1,\ldots,n$
are equal as%
\[
\sum_{l=1}^{n}((d+1)(p_{l}+a\sum_{j\in A_{l}}p_{j})-1)=(d+1)(1+a\left(
n-\varkappa-1\right)  )-n=\frac{n}{d}=\frac{n}{d}\sum_{l=1}^{n}p_{l}\text{.}%
\]
Hence, $(i)$ follows.\smallskip

To show $(i)\Rightarrow(ii)$, we note first that summing the equalities in \eqref{itm:lin} from Theorem \ref{thm2des}
over $l=1,\ldots,n$, we get
\[
\sum_{l=1}^{n}p_{l}=\frac{n}{\left(  d+1\right)  \left( \left(  n-\varkappa-1 \right)a+1 \right)  -n/d}=1\text{.}
\]
Now fix $C\in\mathcal{C}$. Then
\begin{align*}
\frac{n}{d}\sum_{l\in J_C}p_{l} &  =(d+1)(\sum_{l\in J_C}p_{l}+a\sum_{l\in J_C}%
\sum_{j\in A_{l}}p_{j})-d\\
&  =(d+1)\sum_{l\in J_C}p_{l}+a(d+1)\left(  d-\psi\right)  \left(  1-\sum_{l\in
J_C}p_{l}\right)  -d\\
&  =(d+1)\sum_{l\in J_C}p_{l}+(d+1)\left(  1-\sum_{l\in J_C}p_{l}\right)  -d,
\end{align*}
and so \eqref{MUBur} holds, which completes the proof.
\end{proof}

\begin{Ex}
[Two-distance $2$-design in dimension $5$]\label{dim5}This $2$-design consists
of $45$ projections onto vectors $(1,0,0,0,0)$ and $\frac{1}{4}(0,1,\pm
\eta,\pm\eta,\pm1)$ under all cyclic permutations of their coordinates, where
$\eta:=e^{2\pi i/3}$. The $2$-design POVM is then formed by rescaling these
projections by $\frac{1}{9}$. The inner products between
different states of this $2$-design take the values $0$ and $\frac{1}{4}$. As
a two-distance $2$-design, it carries a strongly regular graph
$\operatorname{srg}(45,12,3,3)$. It turns out that among $78$ non-isomorphic
strongly regular graphs with such parameters, our graph is the one generated by
the point graph of the generalised quadrangle of order $(4,2)$, meaning that
the $45$ vectors constituting the $2$-design can be arranged into $27$
orthonormal bases in such a way that every vector belongs to $3$ bases, every
$2$ vectors belong to at most $1$ basis, and for every vector $v$ and every
base $B$ such that $v\notin B$ there exist a unique vector $v^{\prime}$ and a
basis $B^{\prime}$ such that $v^{\prime}\in B^{\prime}$, $v\in B^{\prime}$ and
$v^{\prime}\in B$, see Fig. \ref{GQ42} (adapted from \cite[p. 61]{Pol98} and
\cite[Fig. 1]{San11}). Obviously, the sum of the probabilities over any basis
is equal to $\frac{1}{9}$. But there are $27$ bases and we need just
$45-(25-1)=21$ linear equations to describe the affine space $\mathcal A$. An example of how to get rid of
$6$  equations (respective bases are marked in blue) in order to obtain a
system of $21$ linearly independent equations is presented in Fig. \ref{GQ42}.
\end{Ex}
\begin{figure}[htb]
	\begin{center}
		\includegraphics[scale=0.045]{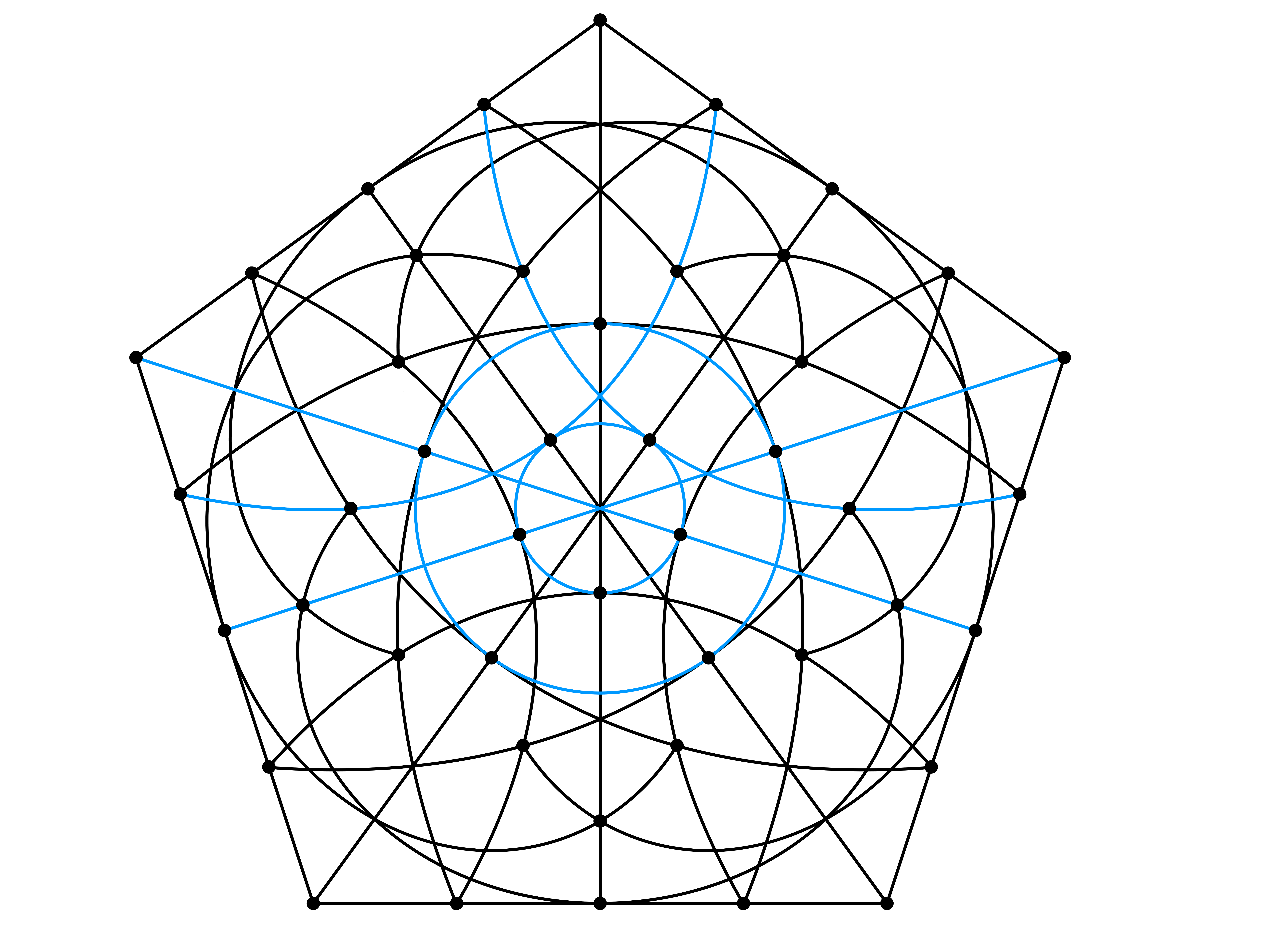}
	\end{center}
	\caption{Two-distance 2-design in $\mathbb C^5$ as $\textnormal{GQ}(4,2)$.
		Vertices represent vectors and curves (lines, circles or arcs of ellipses) represent bases.
		Sums of probabilities over each of blue bases can be omitted in the description of the primal affine subspace
		$\mathcal A$.}
	\label{GQ42}
\end{figure}
\section{Geometric properties of generalised qplex}
\label{Geometric properties of generalised qplex}

To understand a geometric picture hidden behind the algebraic equations
presented in the previous sections, let us consider a $2$-design POVM
$\Pi=(\frac{d}{n}\rho_{j})_{j=1}^{n}$, where $\{
\rho_{j}\}  _{j=1}^{n}\subset\mathcal{P}(\mathbb{C}^{d})$. We show that the geometry of generalised qplex generated by this design is
much alike to that of (Hilbert) qplex, in spite of the fact that now this object is
not located between two simplices but between two dual polytopes lying in a
medial ($d^2-1$)-section of the probability simplex by an affine space. In our reasoning we use similar steps as
in \cite{Appetal17}, although some alterations related to the fact that $\mathcal Q_\Pi$ is not a full-dimensional
subset of $\Delta_n$ are necessary. Finally, we show which geometric properties are preserved if a general morphophoric POVM is considered.

Let us recall that from Corollary \ref{cor2des} it follows that the measurement map
$p_{\Pi}:\mathcal{S}\left(  \mathbb{C}^{d}\right)  \rightarrow\mathcal{Q}%
_{\Pi}$ is a similarity of ratio $m\sqrt{d}$, where $m:=1/\sqrt{n(d+1)}
$. Clearly, the \textit{uniform distribution} $c_{n}:=\left(  1/n,\ldots
,1/n\right)  =p_{\Pi}\left(  \rho_{\ast}\right)  \in\mathcal{Q}_{\Pi}$. We know that
the maximal ball centred at $\rho_{\ast}$ and contained in $\mathcal{S}\left(
\mathbb{C}^{d}\right)  $ has radius $1/\sqrt{d(d-1)}$. Hence, the maximal ball
centred at $c_{n}$ and contained in $\mathcal{Q}_{\Pi}$ has radius
$r:=m/\sqrt{d-1}$. On the other hand, $\mathcal S(\mathbb C^d)$ is contained
in the ball with centre $\rho_{\ast}$ and radius $\sqrt{(d-1)/d}$ with the pure
states contained in the corresponding sphere. Thus, $\mathcal{Q}_{\Pi}$ is contained
in the ball centred at $c_n$ with radius $R:=m\sqrt{d-1}$ and $\mathcal Q_\Pi^1$ is
contained in the corresponding sphere. In short,
\begin{equation}
\label{balls}B_\mathcal A(c_n,r)\subset\mathcal Q_\Pi\subset B_\mathcal A(c_n,R)\quad \textnormal{and}\quad \mathcal Q_\Pi^1\subset \partial B_\mathcal A(c_n,R).
\end{equation}

For $j=1,\ldots,n$ introduce the \textit{basis distributions} given by
$f_{j}:=p_{\Pi}\left(  \rho_{j}\right)  $ generating the \textit{basis
polytope}
$$
D:=\operatorname*{conv}\left\{  f_{j}:j=1,\ldots,n\right\}  \text{.}$$
Note that $\left(  f_{j}\right)_{k}=\left(  f_{k}\right)  _{j}$ for $j,k=1,\ldots,n$ and the \textit{uniform distribution} $c_{n}:=\left(  1/n,\ldots
,1/n\right) = p_{\Pi}\left(  \rho_{\ast}\right)  =p_{\Pi}\left(  \frac{1}%
{n}\sum_{j=1}^{n}\rho_{j}\right) = \frac{1}{n}\sum_{j=1}^{n}f_{j} \in D$.
\begin{F}\label{outball}
$B_{\mathcal{A}}\left(  c_{n},R\right)  $ is the circumscribed ball of $D$.
\end{F}
Now, the following statement is the result of simple computation.

\begin{claim}
\label{affine}Let $p\in\mathbb{R}^{n}$. The following conditions are equivalent:
\begin{enumerate}[i.]
\item $p\in\mathcal{A}$;
\item\label{urg} $p=\left(  d+1\right)  \sum_{j=1}^{n}p_{j}f_{j}-dc_{n}$;
\item $p-c_{n}=\left(  d+1\right)  \sum_{j=1}^{n}p_{j}\left(  f_{j}-c_{n}\right) $;
\item $\frac{1}{d+1}\left\langle p,e_{k}\right\rangle =\left\langle p,f_{k}\right\rangle -dm^{2}$ \ \ for every $k=1,\ldots,n$;
\item\label{orto} $p-c_{n}\perp e_{k}-\left(  d+1\right)  \left(  f_{k}-c_{n}\right)  $
\ \ for every $k=1,\ldots,n$.
\end{enumerate}
\end{claim}

Note that \eqref{urg} is just the famous primal equation
(`Urgleichung') introduced by the founders of QBism
and applied to the situation where $\Pi$ plays the role of both the `ground'
and the `sky' measurement.

Let us now consider  the homothety in $\mathcal{A}$ with the centre at $c_{n}$
and the ratio equal to $1/(d+1)$, i.e. the map $h:\mathcal{A}\rightarrow
\mathcal{A}$ given by
\[
h\left(  p\right)  :=c_{n}+\frac{p-c_{n}}{d+1}=\frac{1}{d+1}p+\frac{d}%
{d+1}c_{n}%
\]
for $p\in\mathcal{A}$. Let $P_{\mathcal{A}}:\mathbb{R}^{n}\rightarrow
\mathcal{A}$ denote the orthogonal projection on $\mathcal{A}$.

\begin{claim}
\label{ortproj}In the above situation $h\left(  P_{\mathcal{A}}e_{k}\right)
=f_{k}$ for every $k=1,\ldots,n$. In particular, $h\left(  P_{\mathcal{A}%
}\Delta_{n}\right)  =D$.
\end{claim}

\begin{proof}
Let $k=1,\ldots,n$. Then it follows from  Fact \ref{affine}.\ref{orto}
that $P_{\mathcal{A}}e_{k}=c_{n}+\left(  d+1\right)  \left(  f_{k}%
-c_{n}\right)  $. Hence, $h\left(  P_{\mathcal{A}}e_{k}\right)  =f_{k}$, as desired.
\end{proof}

Define the \textit{primal polytope} $\Delta$ as the ($d^2-1$)-section of
$\Delta_{n}$ by $\mathcal{A}$, i.e.%
\[
\Delta:=\mathcal{A}\cap\Delta_{n}\text{.}%
\]
It follows from Theorem \ref{thm2des} that%
\[
D\subset\mathcal{Q}_{\Pi}\subset\Delta\text{,}%
\]
see Fig. \ref{Fig4}. Thus%
\[
D\cup B_{\mathcal{A}}\left(  c_{n},r\right)  \subset\mathcal{Q}_{\Pi}%
\subset\Delta\cap B_{\mathcal{A}}\left(  c_{n},R\right)  \text{.}%
\]

\begin{Th}
The above ($d^2-1$)-section is \textit{medial}, i.e. the vertices of $\Delta_{n}$ are
equidistant from $\mathcal{A}$. More precisely,
\[
\operatorname{dist}\left(  e_{k},\mathcal{A}\right)  =\sqrt{1-\frac{d^{2}}{n}}%
\]
for every $k=1,\ldots,n$.
\end{Th}

\begin{proof}
Let $k=1,\ldots,n$. From the Pythagorean theorem, Fact \ref{ortproj}, \eqref{balls}, and Theorem \hyperref[itm:squ]{\ref*{thm2des}.\ref*{itm:squ}} we get
\begin{align*}
\operatorname{dist}^2\left(  e_{k},\mathcal{A}\right)   &  =\left\|
P_{\mathcal{A}}e_{k}-e_{k}\right\|  ^{2}=\left\|  e_{k}-c_{n}\right\|
^{2}-\left\|  P_{\mathcal{A}}e_{k}-c_{n}\right\|  ^{2}\\
&  =\frac{n-1}{n}-\left(  d+1\right)  ^{2}\left\|  f_{k}-c_{n}\right\|
^{2}=\frac{n-1}{n}-\left(  d+1\right)  ^{2}\left(  d-1\right)  m^{2} =1-\frac{d^{2}}{n}\text{.}\qedhere
\end{align*}
\end{proof}
\begin{figure}[htb]	
\begin{center}	\includegraphics[scale=0.7]{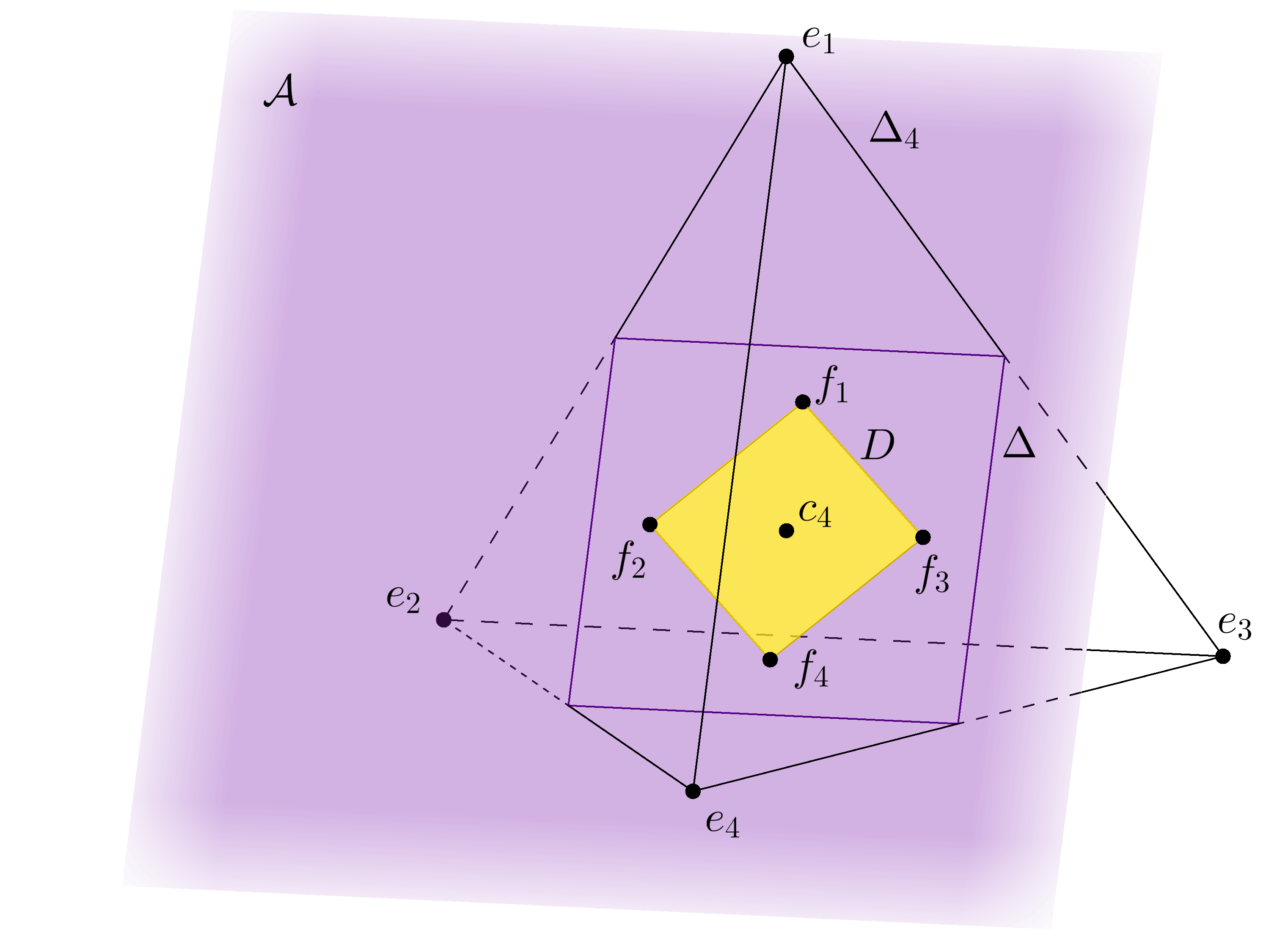}
\end{center}
	\caption{A central and medial 2-section of $\Delta_4$ by the primal affine space $\mathcal A$. The points $f_1,\ldots,f_4$ are images under a homothety in $\mathcal A$ with the centre at $c_4$ of the orthogonal projections of $e_1,\ldots,e_4$ onto $\mathcal A$.}
\label{Fig4}
\end{figure}
Note that $\mathcal{A}\subset\{p\in\mathbb{R}^{n}:\sum_{j=1}^{n}p_{j}=1\}$.
Hence $c_{n}\perp\mathcal{A}$, and so $\left\langle p-c_{n},q-c_{n}%
\right\rangle =\left\langle p,q\right\rangle -\frac{1}{n}$ for $p,q\in
\mathcal{A}$. Let $\iota$ be the inversion in $\mathcal{A}$ through $c_{n}$,
given by $\iota\left(  p\right)  :=2c_{n}-p$ for $p\in\mathcal{A}$. Now we can characterise $\mathcal{A}$ more precisely. The following fact is the result of direct computation.

\begin{claim}
\label{Delta}Let $p\in\mathcal{A}$. The  following conditions are equivalent:

\begin{enumerate}[i.]
\item $p\in\Delta$;

\item $p_{k}\geq0$ $\ \ $for every $k=1,\ldots,n$;

\item $\left\langle p,f_{k}\right\rangle \geq dm^{2}$ $\ \ $for every
$k=1,\ldots,n$;

\item $\left\langle p,q\right\rangle \geq dm^{2}$ $\ \ $for every $q\in D$;

\item \label{Delta5} $m^{2}\geq\left\langle p-c_{n},q-c_{n}\right\rangle $ \ $\ $for every
$q\in\iota\left(  D\right)  $.
\end{enumerate}
\end{claim}

\begin{F} $B_{\mathcal{A}}%
(c_{n},r)$ is the inscribed (maximal central) ball contained in $\Delta$.
\end{F}
\begin{proof}
Taking $\rho_{k}^{\prime
}:=(I-\rho_{k})/(d-1)$ and $f_{k}^{\prime}:=p_{\Pi
}(\rho_{k}^{\prime})$, from the similarity of $\mathcal S(\mathbb C^d)$ and $\mathcal Q_\Pi$ we get  $\left\|  f_{k}^{\prime
}-c_{n}\right\|^2=dm^2\|\rho_k'-\rho_{\ast}\|_{HS}^2=\frac{m^2}{d-1}=r^2$ and $\langle f_k',f_k\rangle=dm^2(\tr(\rho_k'\rho_k)+1)=dm^2$. Thus, from Fact~\ref{Delta} it follows that $f_{k}^{\prime}%
\in\partial_{\mathcal{A}}\Delta\cap B(c_{n},r)$.
\end{proof}

Now we discuss the polarity of basis and primal polytopes.

\begin{definition}
Let $C$ be a convex subset of $\mathcal{A}$, $s>0$. Define \textit{polar} and
\textit{dual} of $C$ in $\mathcal{A}$ with respect to the  sphere with centre at $c$ of radius $s$
by%
\[
C^{\circ}:=\left\{  p\in\mathcal{A}:s^{2}\geq\left\langle p-c%
,q-c\right\rangle \text{ for every }q\in C\right\}
\quad \textnormal{and}\quad
C^{\star}:=\iota\left(  C^{\circ}\right)\text{.}\]
\end{definition}

Note that if the sphere in the definition above is centred at $c_n$, we get $$C^{\star}=\left\{  p\in\mathcal{A}%
:\left\langle p,q\right\rangle \geq \frac{1}{n}-s^{2}\text{ for every }q\in C\right\}
\text{.}$$
From the convex analysis we know that the following fact holds.

\begin{claim}
\label{dual}In the above situation

\begin{enumerate}[i.]
\item $C^{\circ\circ}=C$;

\item $C^{\star}=\left(  \iota\left(  C\right)  \right)  ^{\circ}$;

\item $C^{\star\star}=C$.
\end{enumerate}
\end{claim}

The last statement is a direct consequence of the preceding two. Two convex sets $C$ and $G$ are called \textit{polar} (resp. \textit{dual)
in }$\mathcal{A}$ with respect to the central sphere of radius $s$ if
$C^{\circ}=G$ (resp. $C^{\star}=G$). From condition \eqref{Delta5} in Fact \ref{Delta}
we deduce, see Fig. \ref{Fig5}:

\begin{Th}
The polytopes $D$ and $\Delta$ are \textit{dual }in $\mathcal{A}$ with respect
to the central sphere of radius $m=\sqrt{rR}$. The same is true for the balls
$B_{\mathcal{A}}\left(  c_{n},r\right)  $ and $B_{\mathcal{A}}\left(
c_{n},R\right)  $.
\end{Th}
\begin{figure}[htb]
	\begin{center}
	\includegraphics[scale=0.4]{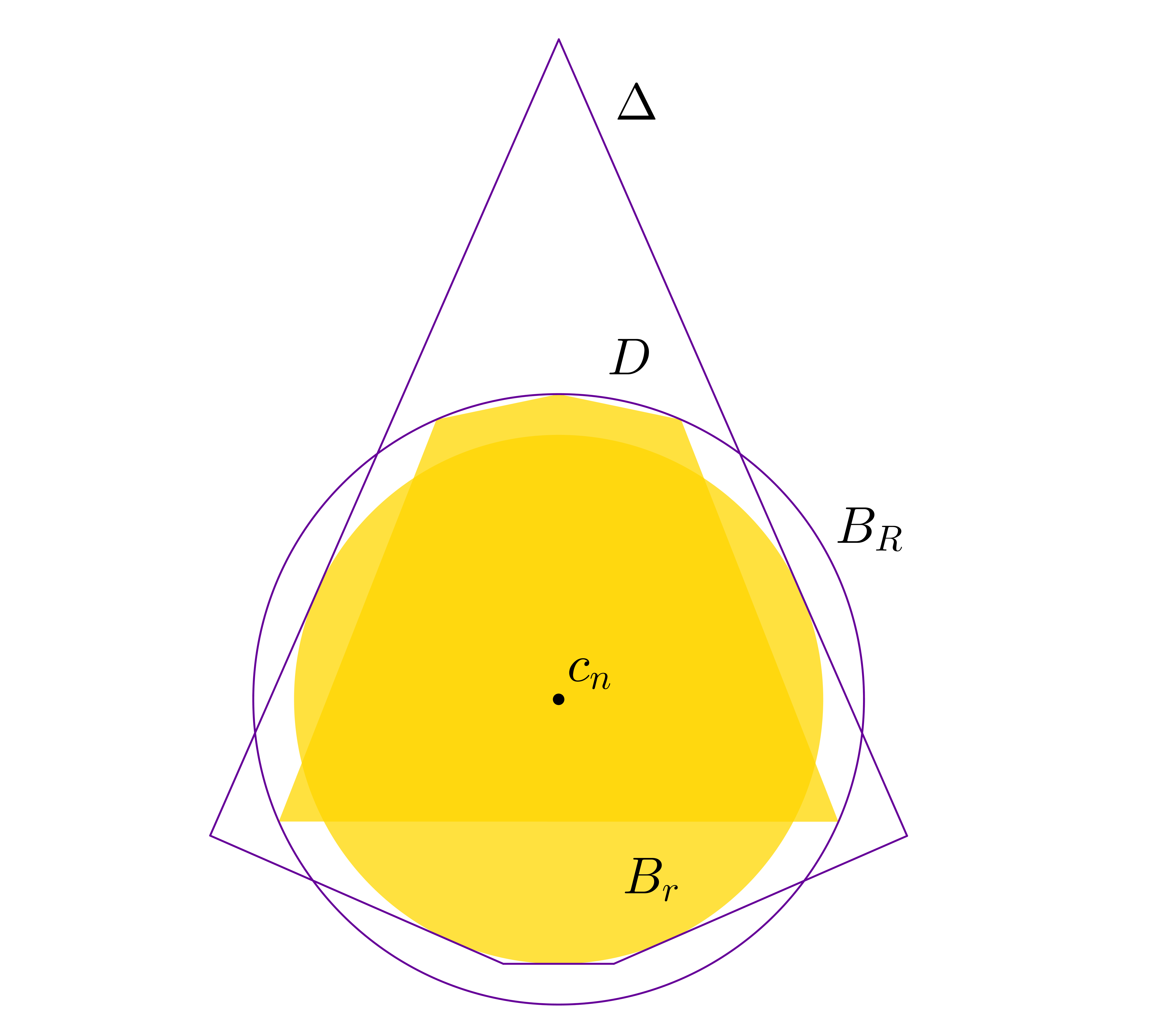}
	\end{center}
	\caption{Duality of the polytopes $D$ and $\Delta$ and of the balls $B_r$ and $B_R$. The qplex $\mathcal Q_\Pi$ is a~convex set containing $D$ and $B_r$ and contained in the intersection of $\Delta$ and $B_R$.}
\label{Fig5}
\end{figure}

Now, we show that the elements of $\mathcal{Q}_{\Pi}$ fulfil the following
\textit{fundamental inequalities} (or are \textit{consistent} \cite{Appetal17}).

\begin{Th}
\label{fund}For $p,q\in\mathcal{Q}_{\Pi}$ we have%
\[
dm^{2}\leq\left\langle p,q\right\rangle \leq2dm^{2}\text{.}%
\]
\end{Th}

\begin{proof} Take $\rho,\sigma\in\mathcal{S}\left(  \mathbb{C}^{d}\right)  $ such that
$p_{\Pi}\left(  \rho\right)  =p$ and $p_{\Pi}\left(  \sigma\right)  =q$.
From the similarity of $\mathcal S(\mathbb C^d)$ and $\mathcal Q_\Pi$ it follows that $\langle p,q\rangle=dm^2(\tr(\rho\sigma)+1)$.
Now, it is enough to apply the well-known inequality $0\leq\operatorname*{tr}%
(\rho\sigma)\leq1$.
\end{proof}

Finally, let us set again $s=m$ in the definition of dual sets. Then $\mathcal{Q}_{\Pi}$ is a self-dual (and hence \textit{maximally consistent}) set.

\begin{Th}
The probability range $\mathcal{Q}_{\Pi}$ is self-dual%
\[
\mathcal{Q}_{\Pi}^{\star}=\mathcal{Q}_{\Pi}\text{.}%
\]
\end{Th}

\begin{proof}
It follows from the first inequality in Theorem \ref{fund} that $\mathcal{Q}%
_{\Pi}\subset\mathcal{Q}_{\Pi}^{\star}$. Let $p\in\mathcal{Q}_{\Pi}^{\star}$
and  $\tau\in\mathcal L_s(\mathbb C^d)$ be such that $\tr\tau=1$ and $p=p_\Pi(\tau)$. Since  \eqref{iso} holds for all $\rho,\sigma\in \mathcal L_s(\mathbb C^d)$ such that $\tr\rho=\tr\sigma=1$, from the definition of duality we get  $\tr(\tau\rho)\geq 0$ for every $\rho\in\mathcal S(\mathbb C^d)$. Hence $\tau\geq 0$ and, in consequence,
$\tau\in\mathcal{S}\left(  \mathbb{C}^{d}\right)  $ and $p\in\mathcal Q_\Pi$.
\end{proof}

It turns out that some of the geometric properties  presented in this section hold true for every morphophoric POVM. The reasoning behind these observations is basically the same as in the case of 2-design POVMs, although some adjustments are necessary. First of all, the balls that sandwich $\mathcal Q_\Pi$ need to be centred at $c_\Pi=(\tr\Pi_1/d,\ldots,\tr\Pi_n/d)$. The radii of these balls are defined in the same way as previously, note, however, that the parameter $m$ involved in the definition is strictly related to the similarity ratio, i.e. $m=\sqrt{\alpha/d}$. Relations \eqref{balls}  now take the following form:
$$B_\mathcal A(c_\Pi,r)\subset\mathcal Q_\Pi\subset B_\mathcal A(c_\Pi,R)\quad \textnormal{and}\quad \mathcal Q_\Pi^1\subset \partial B_\mathcal A(c_\Pi,R).$$
We can also define the basis distributions $f_j=p_{\Pi}(\rho_j)$ ($j=1,\dots,n$) and the basis polytope $D=\operatorname*{conv}\left\{  f_{j}:j=1,\ldots,n\right\}$ setting $\rho_j:=\Pi_j/\tr\Pi_j$. The states $\rho_j$ need not to be pure this time, and so we do not have an analogue of Fact \ref{outball} but we still have $$D\cup B_\mathcal A(c_\Pi,r)\subset\mathcal Q_\Pi\subset\Delta\cap B_\mathcal A(c_\Pi,R).$$

Finally, and most importantly, the appropriate dualities are preserved. Indeed, for $\iota$ denoting now the inversion in $\mathcal A$ through $c_\Pi$ (i.e. $\iota(p):=2c_\Pi-p$) we get the analogue of Fact \ref{Delta} from which we deduce the following properties:
\begin{enumerate}[i.]
	\item The polytopes $D$ and $\Delta$ are dual in $\mathcal A$ with respect to the sphere with centre at $c_\Pi$ of radius $m=\sqrt{\alpha/d}=\sqrt{rR}$. The same is true for the balls $B_\mathcal A(c_\Pi,r)$ and $B_\mathcal A(c_\Pi,R)$.
	\item For $p,q\in\mathcal Q_\Pi$ we have $-m^2\leq\langle p-c_\Pi,q-c_\Pi\rangle\leq R^2$.
	\item The probability range $\mathcal Q_\Pi$ is self-dual ($\mathcal Q_\Pi^\star=\mathcal Q_\Pi$) with respect to the sphere with centre at $c_\Pi$ of radius $m$.
\end{enumerate}

\section{Morphophoric POVM\ and primal equation}

Assume that a morphophoric POVM $\Pi=(\Pi_{j})_{j=1}^{n}$ is given and
$\rho\in\mathcal{S}\left(  \mathbb{C}^{d}\right)  $ is the state of the system
before the measurement. Then the probability that the result of the
measurement $\Pi$ is $j$ ($j=1,\ldots,n$) equals $p_{j}^{\Pi}\left(
\rho\right)  =\operatorname{tr}(\Pi_{j}\rho)$. (In this section we put the POVM in the superscript  and the coordinates in the subscript.) Assume also that the state of the system after the measurement is
described by the \textsl{(generalised) L\"{u}ders instrument }\cite[p.
243]{HeiZim12}, i.e. it is given by $\rho_{j}\left(  \rho\right)  :=\Pi
_{j}^{1/2}\rho\Pi_{j}^{1/2}/p_{j}^{\Pi}\left(  \rho\right)  $ provided the
result was actually $j$ and $p_{j}^{\Pi}\left(  \rho\right)  \neq0$. Let us
now consider another POVM $\Xi=(\Xi_{k})_{k=1}^{N}$. Fuchs \textit{et al}
\cite{FucSta19} called the former measurement \textsl{in the sky} and the
latter \textsl{on the ground}, and analysed two paths leading to the results
of the second measurement. One can either perform the measurement `on
the ground' directly or proceed first to the (counterfactual) measurement `in
the sky' and subsequently to the measurement `on the ground'. In the
former case the probability that the result of the measurement $\Xi$ is $k$
($k=1,\ldots,N$) is given by $p_{k}^{\Xi}\left(  \rho\right)
=\operatorname{tr}\left(  \Xi_{k}\rho\right)  $. In the latter case, the
probability that the results of the measurements $\Pi$ and $\Xi$ are
subsequently $j$ and $k$ ($j=1,\ldots,n$, $k=1,\ldots,N$) is given by
$p_{jk}^{\Pi\Xi}\left(  \rho\right)  =\operatorname{tr}\left(  \Xi_{k}\rho
_{j}\left(  \rho\right)  \right)  p_{j}^{\Pi}\left(  \rho\right)
=\operatorname{tr}\left(  \Xi_{k}\Pi_{j}^{1/2}\rho\Pi_{j}^{1/2}\right)  $ and
can be expressed via the conditional probability $p_{k|j}^{\Xi|\Pi}\left(
\rho\right)  :=\operatorname{tr}\left(  \Xi_{k}\rho_{j}\left(  \rho\right)
\right)  $ as $p_{jk}^{\Pi\Xi}\left(  \rho\right)  =p_{k|j}^{\Xi|\Pi}\left(
\rho\right)  p_{j}^{\Pi}\left(  \rho\right)  $.

The central theme of QBism is answering the question of how
the probabilities $p_{k}^{\Xi}$ depend on $p_{j}^{\Pi}$ and $p_{k|j}^{\Xi|\Pi}$.
Classically, if the process of measurement does not alter the state, such
quantities should be linked by the \textsl{law of total probability}:
$p_{k}^{\Xi}=\sum_{j=1}^{n}p_{j}^{\Pi}p_{k|j}^{\Xi|\Pi}$; however, in the quantum
world the situation is quite different. In the case of rank-$1$ POVMs the
post-measurement states $\rho_{j}$ ($j=1,\ldots,n$) and, in consequence, the
conditional probabilities $p_{k|j}^{\Xi|\Pi}$ do not depend on the
pre-measurement state $\rho$. For SIC-POVMs Fuchs and Schack \cite[p.
23]{FucSch09} used this fact to formulate the fundamental equation of QBism
(called the `primal equation' or the `Urgleichung'):

\begin{equation}
p_{k}^{\Xi}\left(  \rho\right)  =\sum_{j=1}^{n}p_{k|j}^{\Xi|\Pi}\left(
\left(  d+1\right)  p_{j}^{\Pi}\left(  \rho\right)  -1/d\right)
\text{.}\label{Urgleichung SIC}%
\end{equation}
Note, however, that the derivation of such an equation does not require $\Pi$ to be a
SIC-POVM; in fact, any morphophoric rank-$1$ equal-trace POVM, i.e. a $2$-design POVM (see Corollary
\ref{cor2des}), suits here, resulting in the equation of the form%
\begin{equation}
p_{k}^{\Xi}\left(  \rho\right)  =\sum_{j=1}^{n}p_{k|j}^{\Xi|\Pi}\left(
\left(  d+1\right)  p_{j}^{\Pi}\left(  \rho\right)  -d/n\right)
\text{.}\label{Urgleichung rank 1}%
\end{equation}
The proof  is similar to the SIC-POVM case and follows directly from \eqref{2des}.

For an arbitrary morphophoric POVM $\Pi$ the situation could be much more
complicated, as, in general, the post-measurement states for $\Pi$ depend on
the initial state of the system. To derive an analogue of \eqref{Urgleichung
	rank 1}, we have to introduce some notation. Recall that $c_{\Pi}=p_{\Pi
}\left(  \rho_{\ast}\right)$ and $c_{\Xi}=p_{\Xi}\left(  \rho_{\ast}\right)$, where $\rho_{\ast
}:=I/d$ is the maximally mixed state. Consider the \textsl{deviation} of the
probability distribution of the measurement results for the pre-measurement
state $\rho$ from that for $\rho_{\ast}$ for both POVMs:
\begin{align*}
\delta_{\Pi}\left(  \rho\right)   &  :=p_{\Pi}\left(  \rho\right)  -c_{\Pi
}\text{,}\\
\delta_{\Xi}\left(  \rho\right)   &  :=p_{\Xi}\left(  \rho\right)  -c_{\Xi
}\text{,}%
\end{align*}
and the transposed covariance matrix $C$ of two consecutive measurements $\Pi$
and $\Xi$ given by:
\begin{align*}
C_{kj} &  :=p_{jk}^{\Pi\Xi}\left(  \rho_{\ast}\right)  -p_{j}^{\Pi}\left(
\rho_{\ast}\right)  p_{k}^{\Xi}\left(  \rho_{\ast}\right)  =\frac{1}{d}\left(
\operatorname{tr}\left(  \Xi_{k}\Pi_{j}\right)  -\operatorname{tr}\left(
\Xi_{k}\right)  \operatorname{tr}\left(  \Pi_{j}\right)  /d\right)  \\
&  =\operatorname{tr}\left(  \pi_{0}\left(  \Xi_{k}\right)  \pi_{0}\left(
\Pi_{j}\right)  \right)  /d
\end{align*}
for $j=1,\ldots,n$ and $k=1,\ldots,N$, where $\pi_{0}\left(  A\right)
=A-\left(  \operatorname{tr}A\right)  \rho_{\ast}$ for $A\in\mathcal{L}%
_{s}(\mathbb{C}^{d})$, see \eqref{pi0}. Then the following simple equation,
which in the current context replaces \eqref{Urgleichung rank 1}, is fulfilled (see Theorem \ref{genurg}):

\begin{equation}
\delta_{\Xi}=\frac{d}{\alpha}C\delta_{\Pi}\text{,}
\label{Primal equation}%
\end{equation}
where
\begin{equation}
\alpha=\frac{d}{\dim(\mathcal{A})}\sum_{j=1}^{n}\left(  p_{jj}^{\Pi\Pi
}\left(  \rho_{\ast}\right)  -\left(  p_{j}^{\Pi}\left(  \rho_{\ast}\right)
\right)  ^{2}\right)  =\frac{1}{d^{2}-1}\sum_{j=1}^{n}\left(
\operatorname{tr}\left(  \Pi_{j}^{2}\right)  -\frac{\left(  \operatorname{tr}%
	\left(  \Pi_{j}\right)  \right)  ^{2}}{d}\right)  \label{alpha}%
\end{equation}
is the square of the similarity ratio for morphophoric POVM, see \eqref{squaresim}. Note that, as in the original `Urgleichung', all the quantities that appear
in \eqref{Primal equation} can be interpreted in probabilistic terms. In fact,
the above equation fully characterises morphophoric POVMs as the following
theorem shows:

\begin{Th}\label{genurg}
	Assume that $\Pi=(\Pi_{j})_{j=1}^{n}$ is a POVM in $\mathbb{C}^{d}$. Then the
	following conditions are equivalent:
	
	\begin{enumerate}[i.]
		\item $\Pi$ is morphophoric;
		
		\item  for every POVM $\Xi=(\Xi_{k})_{k=1}^{N}$, \eqref{Primal equation} holds
		with $\alpha$ given by \eqref{alpha};
		
		\item  for some IC-POVM $\Xi=(\Xi_{k})_{k=1}^{N}$ and $\alpha>0$, \eqref{Primal equation} holds.
	\end{enumerate}
\label{morpheq}
\end{Th}

\begin{proof}
	The implication $(ii)\Rightarrow(iii)$ is obvious. We prove $(iii)\Rightarrow(i)$. Let
	$k=1,\ldots,N$ and $\rho\in\mathcal{S}\left(  \mathbb{C}^{d}\right)  $. Then
	\begin{align*}
	\alpha\operatorname{tr}\left(  \pi_{0}\left(  \Xi_{k}\right)  \left(
	\rho-\rho_{\ast}\right)  \right)   &  =\alpha\operatorname{tr}\left(  \Xi
	_{k}\left(  \rho-\rho_{\ast}\right)  \right)  \\
	&  =\alpha(\delta_{\Xi}\left(  \rho\right)  )_{k}\\
	&  =\sum_{j=1}^{n}\operatorname{tr}\left(  \pi_{0}\left(  \Xi_{k}\right)
	\pi_{0}\left(  \Pi_{j}\right)  \right)  \left(  p^\Pi_{j}\left(  \rho\right)
	-p^\Pi_{j}\left(  \rho_{\ast}\right)  \right)  \\
	&  =\sum_{j=1}^{n}\operatorname{tr}\left(  \pi_{0}\left(  \Xi_{k}\right)
	\pi_{0}\left(  \Pi_{j}\right)  \right)  \operatorname{tr}\left(  \pi
	_{0}\left(  \Pi_{j}\right)  \left(  \rho-\rho_{\ast}\right)  \right)  \text{.}%
	\end{align*}
	
	The set $\left\{  \rho-\rho_{\ast}:\rho\in\mathcal{S}\left(  \mathbb{C}%
	^{d}\right)  \right\}  $ generates $\mathcal{L}_{s}^{0}(\mathbb{C}^{d})$, and, as $\Xi$ is informationally complete, 
	the same is true for $\left\{  \pi_{0}\left(  \Xi_{k}\right)
	:k=1,\ldots,N\right\}$, see Proposition \ref{IC}. In
	consequence,
	\begin{equation}
	\alpha\operatorname{tr}\left(  AB\right)  =\sum_{j=1}^{n}\operatorname{tr}%
	\left(  A\pi_{0}\left(  \Pi_{j}\right)  \right)  \operatorname{tr}\left(
	\pi_{0}\left(  \Pi_{j}\right)  B\right)  \label{frame pi}%
	\end{equation}
	for every $A,B\in\mathcal{L}_{s}^{0}(\mathbb{C}^{d})$. Now, from Theorem
	\ref{tf} we deduce that $\pi_{0}\left(  \Pi\right)  $ is a tight frame in
	$\mathcal{L}_{s}^{0}(\mathbb{C}^{d})$, and so, by Theorem \ref{thmsim}, we get
	$(i)$. Finally, in the same manner from $(i)$ we get \eqref{frame pi} with $\alpha$
	given by \eqref{alpha}. Let us consider an arbitrary (not necessarily IC) POVM
	$\Xi=(\Xi_{k})_{k=1}^{N}$. Putting $A:=\pi_{0}\left(  \Xi_{k}\right)  $
	($k=1,\ldots,N$) and $B:=\rho-\rho_{\ast}$ ($\rho\in\mathcal{S}\left(
	\mathbb{C}^{d}\right)  $), we can derive \eqref{Primal equation} in the same
	way as above.
\end{proof}

It is an easy computation to show that for rank-$1$ equal-trace POVMs \eqref{Urgleichung rank 1} and \eqref{Primal equation} are equivalent. Hence, and from Corollary
\ref{cor2des}, we get the following result:

\begin{Cor}
	Assume that $\Pi=(\Pi_{j})_{j=1}^{n}$ is a rank-$1$ equal-trace POVM in $\mathbb{C}^{d}$.
	Then the following conditions are equivalent:
	
	\begin{enumerate}[i.]
		\item $\Pi$ is morphophoric;
		
		\item $\Pi$ is a $2$-design POVM;
		
		\item  for every POVM $\Xi=(\Xi_{k})_{k=1}^{N}$ \eqref{Urgleichung rank 1} holds;
		
		\item  for some IC-POVM $\Xi=(\Xi_{k})_{k=1}^{N}$ and $\alpha>0$
		\eqref{Urgleichung rank 1} holds.
	\end{enumerate}
\label{morpheqcor}
\end{Cor}

Note that \eqref{Primal equation} follows from the morphophoricity of the measurement rather than from its quantumness.
Considering a classical fuzzy measurement of a classical system that is morphophoric, we get exactly the same equation.

\section{Conclusions}
We believe that we have correctly identified and characterised morphophoricity as the property of
a~quantum measurement that enables the extension of QBism formalism to the broadest possible class
of measurements. In case of rank-$1$ equal-trace morphophoric POVMs, i.e.
$2$-design POVMs, both algebraic and geometric properties of the resulting generalised
qplexes are similar to the SIC-POVM case. In the general case the situation is more complex.
However, even then we can recover some essential geometrical features of qplexes as well as find a surprising
generalisation of the quantum total law of probability. Thus, we hope that our paper
may become the first step towards an abstract definition of a generalised qplex.

Finally, we would like to mention that figuring out how our qplex
is located inside the overall simplex may provide a way to
solve the problems of minimising and maximising the Shannon entropy of morphophoric POVMs.
This in turn leads to the understanding of the lower and upper bounds for the uncertainty of the
measurement results \cite{Braetal16,Dal14,Daletal11,Oreetal11,Szy14,SloSzy16,SzySlo16,Szy18}.

\section*{Acknowledgements}
We thank Anna Szczepanek for numerous remarks that improve the readability of the paper.
AS is supported by Grant No.\ 2016/21/D/ST1/02414 of the Polish National Science Centre.
WS is supported by Grant No.\ 2015/18/A/ST2/00274 of the Polish National Science Centre.

\end{document}